\documentclass[journal]{IEEEtran}

\usepackage{amsmath}
\usepackage{amsfonts}
\usepackage{amssymb}
\usepackage{amsthm}
\usepackage{tabularx}
\usepackage{mathrsfs} 
\usepackage{soul}

\usepackage{url}
\usepackage{hyperref}
\newtheorem{theorem}{Theorem}

\newtheorem{proposition}{Proposition}

\usepackage{stfloats}
\usepackage{float}
\usepackage{graphicx}
\usepackage{cite}
\usepackage{xcolor}
\usepackage{subfigure}

\usepackage{cleveref}

\usepackage{threeparttable}

\renewcommand{\vec}[1]{\mathbf{#1}}

\makeatletter
\def\blfootnote{\xdef\@thefnmark{}\@footnotetext}
\makeatother

\usepackage{lipsum} 

\begin{document}

\title{\huge{Secure Backscatter Communications Through RIS: \\Modeling and Performance
}} 
	\author{Masoud~Kaveh\IEEEmembership{}, Farshad Rostami Ghadi, \IEEEmembership{Member, IEEE}, Zhao Li, \IEEEmembership{Member, IEEE}, Zheng Yan, \IEEEmembership{Fellow, IEEE},\\ and  Riku Jäntti, \IEEEmembership{Senior Member, IEEE}
	}
	\maketitle
	\vspace{-0pt}
 \begin{abstract}
Backscatter communication (BC) has emerged as a pivotal wireless communication paradigm owing to its low-power and cost-effective characteristics. However, BC faces various challenges from its low signal detection rate to its security vulnerabilities. Recently, reconfigurable intelligent surfaces (RIS) have surfaced as a transformative technology addressing power and communication performance issues in BC. However, the potential of RIS in addressing the security challenges of BC remains uncharted. This paper investigates the secrecy performance of RIS-aided BC, where all channels are distributed according to the Fisher-Snedecor $\mathcal{F}$ distribution. Specifically, we consider a RIS with $N$ reflecting elements to help a backscatter device (BD) establish a smart environment and enhance the secrecy performance in BC. Due to the nature of BC systems, our analysis considers two possible scenarios (i) in the absence of direct links and (ii) in the presence of direct links. In both cases, we first derive compact analytical expressions of the probability density function (PDF) and cumulative distribution function (CDF) for the received signal-to-noise ratio (SNR) at both a legitimate receiver and an eavesdropper. Then, to analyze the secrecy performance, we further derive analytical expressions of the average secrecy capacity (ASC) and secrecy outage probability (SOP) for both mentioned scenarios. In addition, regarding the importance of system behavior in a high SNR regime, we provide an asymptotic analysis of the SOP and ASC. Eventually, the Monte-Carlo simulation is used to validate the analytical results, revealing that utilizing RIS can greatly improve the secrecy performance of the BC system relative to traditional BC setups that do not incorporate RIS.
	\end{abstract}
	\maketitle
 \blfootnote{\noindent  This work is supported in part by the Academy of Finland under Grants 345072 and 350464.}

	 	\blfootnote{\noindent M. Kaveh and R. Jäntti are with the Department of Information and Communication Engineering, Aalto University, Espoo, Finland. (e-mail: $\rm masoud.kaveh@aalto.fi, riku.jantti@aalto.fi$)}
	 	
	 	\blfootnote{\noindent  F. R. Ghadi is with the Department of Electronic and Electrical Engineering, University College London, WC1E
		6BT London, UK. (e-mail: $\rm f.rostamighadi@ucl.ac.uk$).}

   \blfootnote{\noindent Z. Li and Z. Yan are with the School of Cyber Engineering, Xidian University, Xi'an, China, (e-mail: $\rm zli@xidian.edu.cn, zyan@xidian.edu.cn$)}

	

	\section{Introduction}\label{introduction}
\IEEEPARstart{B}{ackscatter} communication (BC) is a wireless communication technique that enables low-power and low-data-rate devices to transmit data by modulating and reflecting existing radio frequency (RF) signals in the environment \cite{Srvy_AmBC_CompNetw2022}. This technology holds exceptional promise, particularly in the context of the Internet of Things (IoT), due to its inherent energy efficiency, cost-effectiveness, and simplicity \cite{BLE_INFOCOM2020}.
Backscatter devices (BDs) are designed to consume minimal energy as they do not need to generate their own signals; instead, they efficiently utilize the ambient RF signals already present in the environment. Therefore, BC allows the battery-less IoT devices operate for extended periods without requiring frequent battery replacements \cite{BC_Battery_SPM2018}.

While BC offers several advantages for IoT, there are some challenges that have yet to be addressed to ensure its reliable operation and widespread adoption. 
These challenges, inherently arising due to BC's low-power operation nature, unlicensed spectrum usage, and reliance on ambient signals, mainly include low signal detection rate, limited transmission range, low data rate and throughput, and concerns related to energy harvesting and consumption \cite{BC_Battery_CST2023}.  
Recently, reconfigurable intelligent surfaces (RIS) \cite{ref9} has shown promising potentials to improve the communication performance 
in BC systems 
\cite{RIS_BC_frontier_6G,RIS_BC_Survey_Proceeding}.
RIS is essentially a metasurface composed of a large number of passive elements,
which takes the advantage of meta-materials to dynamically control and shape the reflection and scattering of RF signals and improve the signal quality in wireless propagation environment \cite{RIS1,RIS2}. 
RIS has demonstrated remarkable efficacy in tackling significant challenges in BC including improving the channel conditions \cite{ref17,ref18}, transmission range \cite{ref20}, system throughput \cite{ref24,ref25,ref26}, energy efficiency \cite{ref29,ref31}, detection performance \cite{ref32,ref33,ref36}, and energy harvesting \cite{ref39,ref41} (see \ref{subsec_1_3}).
These enhancements been observed across various BC systems, including monostatic BC, bistatic BC, and ambient BC \cite{RIS_BC_frontier_6G}.
Moreover, RIS has shown to be harnessed in BC through different approaches i.e., serving as a helper element or a dedicated BD \cite{RIS_BC_Survey_Proceeding}.

Furthermore, BC signals are typically passive and operate in an open broadcast fashion, rendering them susceptible to interception and eavesdropping by malicious entities. As BDs often operate with stringent resource constraints, including limited processing power and memory, their ability to employ sophisticated security protocols and cryptographic primitives is further restricted. 
Therefore, physical layer security (PLS) stands as a compelling approach for establishing secure BC by offering lower complexity and better security features than cryptographic schemes \cite{ref43,ref45,sensors_BC_survey}.
In this regard, examining PLS performance is essential for establishing resilient and efficient secure communication within wireless networks. In this context, two critical PLS performance metrics, namely average secrecy capacity (ASC) and secrecy outage probability (SOP), have been assessed across various scenarios within BC through recent years \cite{Saad_BC_SOP, Zheng_MultiBC_SOP, ref57, Liu_BC_ASCSOP, PLS_AmBC_Muratkar, PLS_AmBC_NOMA, PLS_AmBC_ITS} (see \ref{subsec_1_1}).

\subsection{Research Gaps and Motivations}

While there have been several efforts to enhance PLS performance in BC during recent years \cite{Saad_BC_SOP, Zheng_MultiBC_SOP, ref57, Liu_BC_ASCSOP, PLS_AmBC_Muratkar, PLS_AmBC_NOMA, PLS_AmBC_ITS}, achieving optimal PLS performance in BC systems has been a challenging task for the related works, primarily due to constraints inherent in conventional BC paradigms, such as restricted signal strength, interference, limited channel knowledge, and resource constraints of BDs  \cite{sensors_BC_survey}.
Despite of RIS great enhancement to BC \cite{ref17,ref18,ref20,ref24,ref25,ref26,ref29,ref31,ref32,ref33,ref36,ref39,ref41}, it has been remained unexplored whether RIS can also enhance the PLS performance of BC systems.
Motivated by the suboptimal PLS performance in BC, and the notable benefits that RIS can offer to BC, we are propelled to leverage the great potential of RIS to improve the PLS performance in BC.
As an initial step in this process, developing analytical expressions for secrecy metrics enables a quantitative evaluation of how various system parameters influence PLS performance.
Specifically, our analysis focuses on assessing ASC and SOP as crucial PLS metrics, which to the best of our knowledge, have not been studied in previous works in RIS-aided BC. 

On the other hand, since the impact of multipath fading and shadowing on the received signal strength in BC has not yet been thoroughly explored, despite its substantial significance in evaluating the performance of BC \cite{ref25,ref56}, employing a flexible composite fading channel model, such as Fisher-Snedecor $\mathcal{F}$ \cite{Fisher_source1}, can provide a more precise assessment of secrecy performance.
Moreover, given the inherently weak direct links \footnote{In this context, \textit{direct link} refers to the links between BDs and either backscatter or eavesdropping receivers. As for the links between the RF source and BDs, it is referred to as the \textit{source link}.} in BC, it becomes imperative to encompass both potential scenarios in the secrecy performance evaluation process. 
This includes assessing the performance of RIS-aided BC systems with and without direct links. The inclusion of both scenarios allows for a comprehensive evaluation of BC's effectiveness and the potential enhancements offered by RIS technology \cite{Back_strngr_Direct_RIS}.
The proposed system and channel models have significant implications across a variety of real-world scenarios, e.g. smart homes, industrial applications, healthcare, smart city, and agriculture, where RIS-aided BC enhances connectivity, energy efficiency, monitoring resolution, and network coverage, for low-power IoT devices. The Fisher-Snedecor $\mathcal{F}$ fading channel can also accurately model indoor BC propagation with obstacles, the complex interactions between signals and the varied urban setting, and the shadowing from farm structures and natural terrain \cite{RIS5GCommSenseSec}.

\subsection{Contributions}

In this paper, we investigate the transformative potential of RIS in enhancing the performance of secure BC under Fisher-Snedecor $\mathcal{F}$ fading channels. Our work addresses critical gaps in the understanding of secure BC systems by encompassing a range of important aspects. Specifically, we offer comprehensive insights into the impact of RIS technology on BC's secrecy performance. The main contributions of this paper can be summarized as follows.

\textbullet \ For the first time in this paper, we assess the secrecy performance of the RIS-aided BC system, considering how different system parameters influence its effectiveness. In addition, due to the relatively weaker direct link in practical scenarios for BC, our analysis includes both possible cases, one with RIS-aided links alone and the other with a combination of RIS-aided and direct links.

\textbullet \ We employ the Fisher-Snedecor $\mathcal{F}$ distribution \cite{Fisher_source1} to model the fading channels, allowing us to accurately characterize the simultaneous occurrence of multi-path fading and shadowing in BC, and consequently, to reach a more accurate secrecy performance analysis compared to other fading distributions. 

\textbullet \ We derive compact analytical 
expressions for the probability density function (PDF) and cumulative distribution function (CDF) of the received SNR at the legitimate receiver and the eavesdropper in the RIS-aided BC system for both scenarios (i.e., with and without direct links) with assumption of a perfect channel state information (CSI) of BDs and imperfect CSI of eavesdroppers at the RIS.
Then, by using the derived PDFs and CDFs, we are able to derive accurate analytical expressions of the ASC and SOP based on the bivariate and multivariate Fox's H-function for analyzing the system's secrecy performance.

\textbullet \ We provide an asymptotic analysis of the obtained ASC and SOP by employing the residue method \cite{residal1}, for studying the RIS-aided BC system's behavior in high SNR regime and gaining insights into the fundamental performance trends and capabilities of BC systems under optimal conditions. 

\textbullet \ We validate the analytical results using Monte-Carlo simulation. 
Our findings indicate that RIS can significantly improve the PLS performance in BC systems across diverse system configurations.
Furthermore, utilizing the Fisher-Snedecor $\mathcal{F}$ distribution can accurately model channel and offer a more precise secrecy performance analysis in comparison to other fading distributions in BC.

\vspace{0 pt}

\subsection{Organization and Notations}

The rest of this paper is organized as follows. 
Section \ref{reWork} delves into the existing literature on the subject.
Section \ref{system-model} presents the RIS-aided BC system and channel models. Section \ref{SNR Distribution} demonstrates the SNR distributions at the legitimate receiver and eavesdropper. We analyze the secrecy performance of RIS-aided BC by deriving the compact analytical expressions of ASC and SOP in Section \ref{sop_section}. Section \ref{sec_asy} presents the asymptotic analysis of the secrecy metrics. The simulation results are discussed in Section \ref{num-results}, and finally a conclusion is drawn in the last section. 


\textit{Notations}: 
$\Gamma(.)$ is the complete Gamma function \cite[Eq. 8.31]{table_int}, $B(.,.)$ is the Beta function \cite[Eq. 8.38]{table_int}, $G^{m,n}_{p,q}(.)$ is the Meijer's G-function \cite[Eq. 8.2.1.1]{ref59}, $H^{m,n:m_1,n_1;...;m_r,n_r}_{p,q:p_1,q_1;...;p_r,q_r}(.)$ is the multivariate Fox's H-function \cite{ref61}, $j=\sqrt{-1}$, and $\vec{X}^\mathcal{T}$ is the transpose of $\vec{X}$.

\section{Related Works} \label{reWork}
In this section, we present a review of the literature concerning performance analysis frameworks for PLS in BC, strategies for enhancing PLS performance through RIS, and the utilization of RIS in BC systems.
Table \ref{table_1} shows the differences between previous studies and our research, highlighting the distinctive aspects of our analysis in this paper.

\begin{table} [h]
    \centering
    \begin{threeparttable}
    \caption{Comparison of Related Works: RIS-Aided BC Systems and PLS Performance Analysis in BC versus Our Work} 
    \label{table_1}
    {
    \begin{tabular}{c|ccccc}
    \hline \hline
      Works   & $I_1$ & $I_2$ & $I_3$ & $I_4$ & BC Performance Evaluation Metrics \\
    \hline
    [10], [11] & \checkmark & $\times$ & $\times$ & $\times$ & Channel Condition\\
    \hline
      [12] & \checkmark & $\times$ & $\times$ & $\times$ & Transmission Range \\
    \hline
        [13]–[15] & \checkmark & $\times$ & $\times$ & $\times$ & System Throughput \\
    \hline
        [16], [17] & \checkmark & $\times$ & $\times$ & $\times$ & Energy Efficiency \\
    \hline
        [18]–[20] & \checkmark & $\times$ & $\times$ & $\times$ & Bit Error Rate \\
    \hline
        [21], [22] & \checkmark & $\times$ & $\times$ & $\times$ & Energy Harvesting  \\
         \hline
        [26]–[28] & $\times$ & \checkmark & $\times$ & $\times$ & SOP \\
        \hline
        [29]–[32] & $\times$ & \checkmark & $\times$ & \checkmark & SOP, ASC  \\
        \hline
        [44] & $\times$ & \checkmark & $\times$ & \checkmark & SOP  \\
    \hline
        Ours & \checkmark & \checkmark & \checkmark & \checkmark & SOP, ASC \\
    \hline
    \end{tabular}
    } 
    \begin{tablenotes}
        \item $I_1$: RIS integration into BC, $I_2$: PLS performance evaluation, $I_3$: Considering effect of multipath fading and shadowing, $I_4$: Asymptotic analysis, \checkmark: Item is supported, $\times$: Item is not supported.
    \end{tablenotes}
\end{threeparttable}
\end{table}

\subsection{PLS Performance Analysis for BC Systems} \label{subsec_1_1}

In recent years, there has been a notable body of work dedicated to assessing the secrecy performance of BC systems \cite{Saad_BC_SOP, Zheng_MultiBC_SOP, ref57, Liu_BC_ASCSOP, PLS_AmBC_Muratkar, PLS_AmBC_NOMA, PLS_AmBC_ITS}.
In \cite{Saad_BC_SOP}, researchers conducted an examination of the performance of wireless backscatter systems, specifically focusing on the evaluation of SOP.
Additionally, \cite{Zheng_MultiBC_SOP} delved into the analysis of SOP within the context of a multi-tag BC system, considering the presence of an eavesdropper.
Furthermore, in \cite{ref57}, the authors improved SOP by introducing an optimal tag selection scheme for passive BC systems characterized by multiple tags and a single eavesdropper.
In \cite{Liu_BC_ASCSOP}, a novel tag selection scheme was proposed with the aim of enhancing both the ASC and SOP of a multi-tag self-powered BC system in scenarios involving an eavesdropper. The authors also derived an analytical expression for SOP to facilitate their analysis.

In \cite{PLS_AmBC_Muratkar}, the investigation centered around the influence of eavesdroppers and the motion of readers on the secrecy performance of ambient BC systems, particularly in cases where channel estimation is imperfect.
In \cite{PLS_AmBC_NOMA}, researchers introduced an overlay cognitive ambient BC non-orthogonal multiple access (NOMA) system tailored for intelligent transportation systems. Within this context, they scrutinized the secrecy performance of their proposed system model in the presence of an eavesdropping vehicle by deriving the SOP.
Additionally, in \cite{PLS_AmBC_ITS}, an exploration was conducted into secure multi-antenna transmission within ambient BC-based intelligent transportation systems. This investigation was carried out in the presence of a passive eavesdropper with jamming, where a cooperative jammer was strategically positioned within the system to disrupt the eavesdropper without affecting the reader. To assess the performance of this proposed scheme, a new closed-form expression for SOP was derived.
However, attaining an optimal PLS performance has proven to be a formidable endeavor in prior research. This difficulty can largely be ascribed to the inherent limitations of traditional BC systems \cite{sensors_BC_survey}.

\subsection{Enhancing PLS Performance Using RIS} 

Through recent years, RIS technology has been offering a promising avenue to enhance the PLS performance in various wireless communication systems. 
These include but are not limited to smart grid communications \cite{ref77}, vehicular networks \cite{RIS_PLS_1}, device-to-device communications \cite{RIS_PLS_2}, IoT networks \cite{RIS_PLS_3}, integrated satellite-vehicle networks \cite{RIS_PLS_4}, and networks involving unmanned aerial vehicles (UAVs) \cite{RIS_PLS_5}.
In addition, by assuming an ambient BC system, the authors in \cite{RIS_PLS_AmBC_Green} derived SOP for a RIS-aided set-up; however, they did not consider the source link (i.e., product channels) in their performance analysis. 
Tang et al. introduced an innovative RIS design aimed at enhancing PLS for RIS-aided NOMA networks, showcasing the potential of RIS to improve security measures dynamically \cite{RISPLS1}. Similarly, Zhang et al. explored the general benefits of RIS in enhancing PLS, providing robust strategies against eavesdroppers in diverse network conditions \cite{RISPLS2}.
Gu et al. addressed the challenge of uncertain eavesdropper locations by employing RIS in their security strategies, demonstrating the flexibility of RIS in adapting to varying security requirements and enhancing overall network resilience \cite{RISPLS3}. In a similar vein, Zhang et al. focused specifically on the integration of RIS into 6G networks using NOMA, illustrating significant improvements in secrecy performance due to the smart deployment of RIS elements \cite{RISPLS4}.
Extending the application of RIS, Wang et al. investigated the uplink secrecy performance of RIS-based radio frequency/free space optics (RF/FSO) three-dimensional heterogeneous networks, highlighting the versatility of RIS in different transmission mediums and its effectiveness in securing communications across multiple layers \cite{RISPLS5}.
Elhoushy et al. utilized RIS to limit information leakage in a cell-free massive multiple input multiple output (MIMO) setting, facing active eavesdroppers. This study particularly noted the efficacy of RIS in environments where traditional security measures fall short \cite{RIS_Leak3}.

\vspace{0pt}
\subsection{RIS-Aided BC Systems}
\label{subsec_1_3}

Quite recently, integrating RIS into BC has led to significant enhancements in different aspects.
The authors in \cite{ref17,ref18} demonstrated that RIS is able to address channel condition challenges observed in conventional BC systems. These challenges arise from various propagation effects, resulting in the arrival of multiple out-of-phase signals at the BD and reader, ultimately leading to a decline in system performance. However, through intelligent signal reflection, RIS effectively mitigates the adverse impacts of electromagnetic radiation and enhances the overall channel gain. 
RIS can also tackle the limited transmission range issue in BC systems by introducing efficient supplementary paths. This capability facilitates the widespread deployment of BC and prevents its utilization from being restricted to short-range applications \cite{ref20}.

Due to the shared transmission medium among BDs in BC, system throughput can be constrained by interference and collisions between different devices. To address this challenge, researchers in \cite{ref24,ref25,ref26} leveraged RIS in various BC scenarios, aiming to improve the quality of service (QoS) and weighted sum rate. This was achieved through optimization of the RIS phase-shift matrix and beamforming vectors at the transmitter.
In addition, RIS's passive reflections offer power gains, allowing RIS-aided BC to leverage this advantage for achieving higher performance gains while requiring less transmit power. As a result, RIS significantly enhances the energy efficiency of BC systems \cite{ref29,ref31}.
In BC systems, the reader's signal detection capability is often hindered by factors like direct-link interference. The authors in \cite{ref32,ref33,ref36} utilized the RIS capabilities to control signal direction, effectively mitigating direct-link interference and enhancing signal detection performance with reduced complexity compared to conventional schemes.
Furthermore, the authors in \cite{ref39,ref41} demonstrated that employing RIS for the coherent combination of reflected signals results in a substantial increase in the total received power within BC systems. This enhancement enables BDs to harvest energy from both direct and RIS-reflected signals, significantly boosting the total harvested energy and supporting long-term IoT network operation.

\vspace{-0pt}
	\section{System Model}\label{system-model}

Fig. \ref{fig:sysmodel} shows the system model of RIS-aided BC. In this model, we consider a BD (like a tag) as a semi-passive device which is powered up through a continuous wave carrier signal transmitted by a source. 
The tag aims to send its confidential messages to a legitimate receiver (like a reader) with the aid of an RIS that has \emph{N} reflecting elements.
We assume that there is also a passive eavesdropper named Eve that tries to decode the confidential message sent by the tag over both tag-to-reader and RIS-to-reader links. 
For simplicity and without loss of generality, since the RF source is sending unmodulated carriers, both reader and Eve can utilize cancellation methods to mitigate the impact of interference due to the source's link \cite{BLE_INFOCOM2020,Riku_DirPathInterf, Riku_IEEERFID1}, which is also a prevalent assumption in RIS-aided BC \cite{ref25,ref56}.
We assume that the CSI of cascade links are known by RIS, so RIS can adjust the phase shifting coefficients of its elements for maximizing the SNR at the reader \footnote{This process can be executed at the reader, which involves transmitting pilot symbols from BDs through different phase configurations of the RIS. This allows for a linear estimation of the cascaded channel by aggregating the received signals that correspond to each configuration. Then, the estimated channel coefficients are communicated back to the RIS controller \cite{BjornsonCSI}.}. 
We also assume that the tag, reader, and Eve are equipped with a single antenna for simplicity. Therefore, the received signal at the tag can be given as follows
\begin{figure}[t]
    \centering    \includegraphics[width=0.49\textwidth]{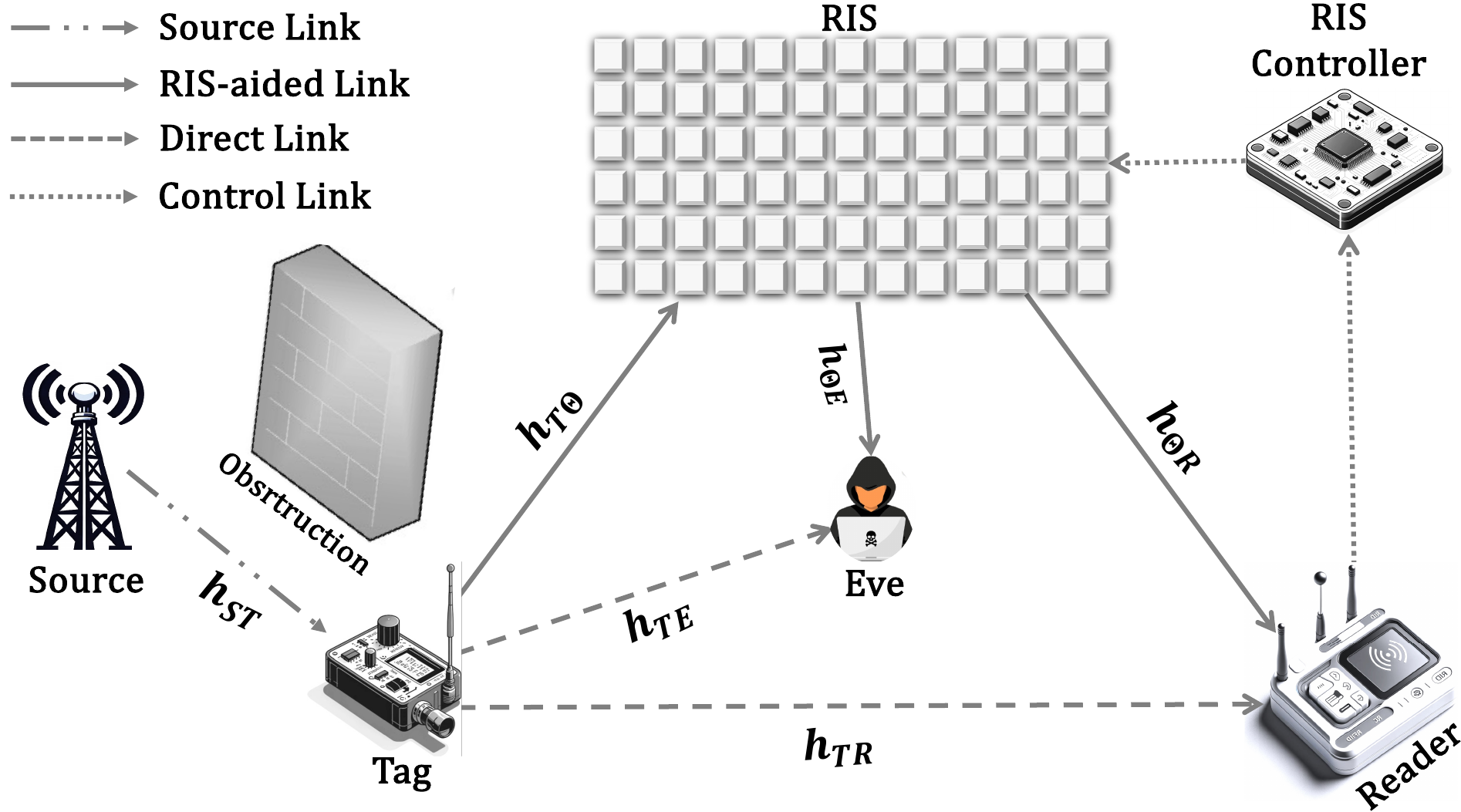}
    \caption{The system model of RIS-aided BC.}
    \label{fig:sysmodel}
\end{figure}
	\begin{align} 
y_T=\sqrt{P_s}h_{ST}+n_T,
\end{align} 
where $P_s$ represents the transmit power of the source, $h_{ST}$ is the channel coefficient between source and tag, and $n_T$ denotes the additive white Gaussian noise (AWGN) at the tag. Since the noise power caused by the tag's antenna is considerably smaller than the received signal from the source \cite{ref57}, we neglect it in the rest of this paper. 

In most practical
scenarios, the direct link between the BD and reader may not be feasible due to distance, obstructions, or channel conditions in BC.
Additionally, the use of low-power BDs further contributes to the limited signal strength and reduces the direct link's reliability. In such cases, RIS can act as the primary means of communication between the tag and reader \cite{Back_strngr_Direct_RIS}. In other cases where a reliable link exists between the BD and reader, RIS can serve as an enhancer, further optimizing signal strength and improving overall system capacity. Therefore, we divide our analysis into two scenarios: RIS-aided BC with and without direct links.
\subsection{Without Direct Links}
Assuming the direct links between tag-to-reader and tag-to-Eve is not feasible, the received signals at the reader and Eve can be expressed as
\begin{align} \label{eq-yr}
y_{R}=\sqrt{P_s}\left(h_{ST}\vec{H}_{T\Theta}\vec{\Theta}\vec{H}^{\mathcal{T}}_{\Theta R}\right)S(t)+n_R,
\end{align}
\begin{align} \label{eq-ye}
y_{E}=\sqrt{P_s}\left(h_{ST}\vec{H}_{T\Theta}\vec{\Theta}\vec{H}^{\mathcal{T}}_{\Theta E}\right)S(t)+n_E,
\end{align}
in which $S(t)$ is the information signal backscattered from the tag with a unit power, $n_R$ and $n_E$ represent the AWGN  at the reader and Eve with zero mean and variances $\sigma^2_R$ and $\sigma^2_E$, respectively, and $\vec{\Theta}$ is the adjustable phase matrix induced by the reflection of RIS elements, which is defined as $\vec{\Theta}=\text{diag}\left(\left[e^{j\theta_1}, e^{j\theta_2},...,e^{j\theta_N}\right]\right)$. 
The vectors $\vec{H}_{T\Theta}$, $\vec{H}_{\Theta R}$, and $\vec{H}_{\Theta E}$ contain the \emph{N} channel coefficients from the tag to
 the RIS, and from RIS to the reader and Eve, respectively.
The above channel vectors are given as $\vec{H}_{T\Theta}=d_{T\Theta}^{-\chi}.\left[h_{T{\Theta}_1}\mathrm{e}^{-j\alpha_1}, h_{T{\Theta}_2}\mathrm{e}^{-j\alpha_2},..., h_{T{\Theta}_N}\mathrm{e}^{-j\alpha_N}\right]$, $\vec{H}_{\Theta R}=d_{\Theta R}^{-\chi}.\left[h_{{\Theta R}_1}\mathrm{e}^{-j\beta_{1}}, h_{{\Theta R}_2}\mathrm{e}^{-j\beta_{2}},..., h_{{\Theta R}_N}\mathrm{e}^{-j\beta_{N}}\right]$, and $\vec{H}_{\Theta E}=d_{\Theta E}^{-\chi}.\left[h_{{\Theta E}_1}\mathrm{e}^{-j\epsilon_{1}}, h_{{\Theta E}_2}\mathrm{e}^{-j\epsilon_{2}},..., h_{{\Theta E}_N}\mathrm{e}^{-j\epsilon_{N}}\right]$, where $d_{T\Theta}$ denotes the distance between the tag and the RIS, $d_{\Theta R}$ is the distance between the RIS and the reader, and $d_{\Theta E}$ defines the distance between the RIS and Eve, respectively. 
The term $\chi$ indicates the path-loss exponent. 
Furthermore, the terms $h_{T{\Theta}_n}$, $h_{{\Theta R}_n}$, and $h_{{\Theta E}_n}$, for $n\in\left\{1,2,...,  N\right\}$, are the amplitudes of the corresponding channel coefficients, and $\mathrm{e}^{-j\alpha_n}$, $\mathrm{e}^{-j\beta_{n}}$, and $\mathrm{e}^{-j\epsilon_{n}}$ denote the phase of the respective links. 
In order to precisely capture the coexistence of multi-path fading and shadowing in BC and achieve a more precise evaluation of secrecy performance, we utilize the Fisher-Snedecor $\mathcal{F}$ distribution \cite{Fisher_source1} as a means to accurately model and characterize the system behavior in our analysis.

\subsection{With Direct Links}
Suppose the direct links between tag-to-reader and tag-to-Eve are existed, thereby, the received signals at the reader and Eve can be given by
\begin{align} \label{Reader_SNR}
y_{R}=\sqrt{P_s}S(t)\left(h_{ST}h_{TR}+h_{ST}\vec{H}_{T\Theta}\vec{\Theta}\vec{H}^{\mathcal{T}}_{\Theta R}\right)+n_R,
\end{align}
\begin{align} \label{Eve_SNR}
y_{E}=\sqrt{P_s}S(t)\left(h_{ST}h_{TE}+h_{ST}\vec{H}_{T\Theta}\vec{\Theta}\vec{H}^{\mathcal{T}}_{\Theta E}\right)+n_E,
\end{align}
where $h_{TR}$ and $h_{TE}$ denote the tag-to-reader and tag-to-Eve channel coefficients, respectively.

\section{SNR Distribution}\label{SNR Distribution}
In this section, an analysis is conducted on the SNR at the reader and Eve by considering both without and with direct link cases. The compact analytical expressions of PDF and CDF are then further derived based on the received SNR.

\subsection{Without Direct Links}

\subsubsection{Legitimate link} 
From \eqref{eq-yr}, the instantaneous
SNR at the reader can be determined as 
\begin{align} 
\gamma_R&=\frac{{\left|\sqrt{P_s}h_{ST}\vec{H}_{T\Theta}\vec{\Theta} \vec{H}_{\Theta R})\right|}^2}{n_R}  \\ 
&=\frac{P_s|h_{ST}|^2\left|\sum_{n=1}^{N}h_{T{\Theta}_n} h_{\Theta {R}_n} \mathrm{e}^{j\left(\theta_n-\alpha_n-\beta_{n}\right)}\right|^2}{{d}^\chi_{ST}{d}^\chi_{T\Theta}{d}^\chi_{\Theta R}{\sigma}^2_R}  \\ 
&\overset{(a)}{=}
\bar{\gamma}_{R}|h_{ST}|^2\left|\sum_{n=1}^{N}h_{T{\Theta}_n} h_{\Theta {R}_n}\right|^2 ,\label{gamma_Reader}
\end{align}
where $(a)$ 
is obtained by enabling ideal phase shifting for RIS \cite{ref41,ref50,ref_csiBC2}, and  $\bar{\gamma}_{R}$ is the average SNR at the reader due to the RIS-aided link.
 By defining \begin{math}X_1=|h_{ST}|^2\end{math} and \begin{math}Y_1=|\sum_{n=1}^{N}h_{T{\Theta}_n} h_{\Theta {R}_n}|^2\end{math}, where all the channels follow Fisher-Snedecor $\mathcal{F}$ fading model, we will have $f_{X_1}(x_1)$ \cite{Fisher_PHY} and $f_{Y_1}(y_1)$ \cite{Farshad_Fisher_RIS} as
 \begin{align} \label{f_{X_1}(x_1)}
 f_{X_1}(x_1)=\mathcal{C}G_{1,1}^{1,1}\left(\begin{array}{c}
				\lambda_{1} x_1\end{array}
			\Big\vert\begin{array}{c}
				-m_{S_{ST}}\\
				m_{ST}-1\\
			\end{array}\right),
 \end{align}
 \begin{align}  \label{f_{Y_1}(y_1)}
f_{Y_1}(y_1)= \frac{{y_1}^{\frac{c-1}{2}}\ \mathrm{e}^{-\frac{\sqrt{y_1}}{{\bar{y}_1}^{\frac{1}{d}}} }  } {2 \hspace{+2pt} {\bar{y_1}}^\frac{c+1}{2} \Gamma(c+1)\ d^{c+1}},
\end{align}
where $m_{ij}$ and $m_{S_{ij}}$ indicate the fading severity
parameter and  the amount of shadowing of the root-mean-square (rms) signal power
parameters, respectively, $\lambda_{1}=\frac{m_{ST}\sigma^2_T}{m_{S_{ST}}P_S}$, $\mathcal{C}=\frac{\lambda_1}{\Gamma\left(m_{ST}\right) \Gamma\left(m_{S_{ST}}\right)}$, 
$c=\frac{(N+1)B'^2 -A'C'}{A'C'-B'^2}$, $d=\frac{D'\left(A'C'-B'^2 \right)}{B'C'}$,  $A'=B\left(m_{\theta R} +1, m_{S_{\theta R}}-1 \right) B\left(m_{T\theta} +1, m_{S_{T\theta}}-1 \right)$, $B'=B\left(m_{\theta R} +\frac{1}{2}, m_{S_{\theta R}} -\frac{1}{2} \right) B\left(m_{T\theta} +\frac{1}{2}, m_{S_{T\theta}} -\frac{1}{2} \right) $, $C'=B\left(m_{\theta R}, m_{S_{\theta R}} \right) B\left(m_{T\theta}, m_{S_{T\theta}} \right) $, $D'=\sqrt{\frac{(m_{S_{\theta R}}-1)(m_{S_{T\theta}}-1) \Omega_{\theta R} \Omega_{T\theta}} {m_{\theta R} m_{T\theta}}}$, and $\Omega_{ij}$ is the mean power.

\begin{theorem} \label{theor_pdf_R}
Assuming all channels follow the Fisher-Snedecor $\mathcal{F}$ fading distribution, the PDF and CDF of $\gamma_R$ without direct links are given by
\begin{align} \label{f_{GR1}(GR1)}
 f_{\gamma_R}(\gamma_R)=\mathcal{G}G^{1,3}_{3,1}\left( \hspace{-6pt} \begin{array}{c}	\frac{4\lambda_{1}\gamma_R}{{\bar{y}_1}^{\frac{2}{d}}}   \end{array} \hspace{-2pt} 		\Big\vert \hspace{-2pt} \begin{array}{c} \frac{2-c}{2}, \frac{3-c}{2}, 2-m_{ST} \\ 1+m_{S_{ST}}\\ \end{array} \hspace{-6pt} \right),
 \end{align}
\begin{align} \label{F_{GR1}(GR1)}
 F_{\gamma_R}(\gamma_R)\hspace{-2pt}=\hspace{-2pt}\mathcal{G} \gamma_R G^{1,4}_{4,2}\left( \hspace{-6pt} \begin{array}{c}	\frac{4\lambda_{1}\gamma_R}{{\bar{y}_1}^{\frac{2}{d}}} \end{array} \hspace{-2pt} 		\Big\vert \hspace{-2pt} \begin{array}{c} 0, \frac{2-c}{2}, \frac{3-c}{2}, 2-m_{ST} \\ 1+m_{S_{ST}}, -1\\ \end{array} \hspace{-6pt} \right),
 \end{align}
where $\mathcal{G}=\frac{ 2^{\frac{2c-3}{2}}{\bar{y}_1}^\frac{c-1}{d} \mathcal{C}}{\sqrt{2\pi} \bar{\gamma}_R {\bar{y}_1}^\frac{c+1}{2} d^{c+1} \Gamma(c+1)}$ and $\bar{y}_1=\frac{P_S}{{d}^\chi_{T\Theta}{d}^\chi_{\Theta R}{\sigma}^2_R}$.
\end{theorem}
\begin{proof} The proof is elaborated in Appendix A.
\end{proof}

\subsubsection{Eavesdropper link} From \eqref{eq-ye}, the instantaneous
SNR at Eve can be determined as 
\begin{align} 
\gamma_E&=\frac{{\left|\sqrt{P_s}h_{ST}\vec{H}_{T\Theta}\vec{\Theta} \vec{H}_{\Theta E})\right|}^2}{n_E}  \\ 
&=\frac{P_s|h_{ST}|^2\left|\sum_{n=1}^{N}h_{T{\Theta}_n} h_{\Theta {E}_n} \mathrm{e}^{j\left(\theta_n-\alpha_n-\epsilon_{n}\right)}\right|^2}{{d}^\chi_{ST}{d}^\chi_{T\Theta}{d}^\chi_{\Theta E}{\sigma}^2_E}  \\ 
&=
\bar{\gamma}_{E}|h_{ST}|^2 \left|\sum_{n=1}^{N}h_{T{\Theta}_n} h_{\Theta {E}_n} \mathrm{e}^{j\left(\theta_n-\alpha_n-\epsilon_{n}\right)}\right|^2 , \label{gamma_Eve}
\end{align}
where $\bar{\gamma}_{E}$ is the average SNR at Eve due to the RIS-aided link. 
When the phase shifts of RIS elements are optimally designed based on the legitimate link's conditions, the resulting phase distributions for each of the Eve's links ($\vec{H}_{T\Theta}\vec{\Theta} \vec{H}^{\mathcal{T}}_{\Theta E}$) are uniformly distributed \cite{PHi_Error_Javier}.
If the phase shift errors are uniformly distributed within the range of [$-\pi$, $\pi$), indicating a complete lack of knowledge about the phases of the RIS, then the channel coefficient follows a circularly-symmetric complex normal distribution. This implies that the equivalent channel exhibits similarities to Rayleigh fading \cite{PHi_Error_1,SOP_Discrete_PHi}.
Therefore, by assuming $Y_2=\left|\sum_{n=1}^{N}h_{T{\Theta}_n} h_{\Theta {E}_n} \mathrm{e}^{j\left(\theta_n-\alpha_n-\epsilon_{n}\right)}\right|^2$, $f_{Y_2}(y_2)$ can be shown as 
\begin{align} \label{f_{Y_2}(y_2)}
f_{Y_2}(y_2)= \frac{1}{a} \  \mathrm{e}^{-\frac{y_2}{a}},  
\end{align}
where $a=\frac{N P_S}{{d}^\chi_{T\Theta}{d}^\chi_{\Theta E}{\sigma}^2_E}$. Now, by considering \eqref{f_{X_1}(x_1)} and \eqref{f_{Y_2}(y_2)}, the marginal distributions of $\gamma_E$ can be obtained as the following theorem. 

\begin{theorem} \label{theor_pdf_E}
Assuming all channels follow the Fisher-Snedecor $\mathcal{F}$ fading distribution, the PDF and CDF of $\gamma_E$ without direct links are given by
\begin{align} \label{f_{GE1}(GE1)}
 f_{\gamma_E}(\gamma_E)= \frac{\mathcal{C}}{a\bar{\gamma}_E} G^{1,2}_{2,1}\left( \hspace{-6pt} \begin{array}{c}	a\lambda_{1}  \gamma_E \end{array} \hspace{-2pt} 		\Big\vert \hspace{-2pt} \begin{array}{c} 1,  2-m_{ST} \\ 1+m_{S_{ST}}\\ \end{array} \hspace{-6pt} \right),
 \end{align}
\begin{align} \label{F_{GE1}(GE1)}
 F_{\gamma_E}(\gamma_E)= \frac{\mathcal{C}}{a\bar{\gamma}_E} \gamma_E G^{1,3}_{3,2}\left( \hspace{-6pt} \begin{array}{c}	a\lambda_{1}  \gamma_E \end{array} \hspace{-2pt} 		\Big\vert \hspace{-2pt} \begin{array}{c} 0, 1,  2-m_{ST} \\ 1+m_{S_{ST}} , -1\\ \end{array} \hspace{-6pt} \right).
 \end{align}    
\end{theorem}

\begin{proof} The proof is elaborated in Appendix B.
\end{proof}

\subsection{With Direct Links}
\subsubsection{Legitimate Link}
According to \eqref{Reader_SNR}, the instantaneous SNR at the reader can be determined as
\newcommand\myeq{\mathrel{\stackrel{\makebox[0pt]{\mbox{\normalfont\small def}}}{=}}}
\newcommand\myapprox{\mathrel{\stackrel{\makebox[0pt]{\mbox{\normalfont\footnotesize (a)}}}{\approx}}}
\begin{align} 
\gamma_R&=\frac{{\left|\sqrt{P_s}(h_{ST}h_{TR}+h_{ST}\vec{H}_{T\Theta}\vec{\Theta} \vec{H}_{\Theta R})\right|}^2}{n_k} \nonumber \\ 
\approx&\frac{P_s|h_{ST}|^2|h_{TR}|^2}{{d}^\chi_{ST}{d}^\chi_{TR}{\sigma}^2_R}
\hspace{-2pt}+\hspace{-2pt}
\frac{P_s|h_{ST}|^2\left|\sum_{n=1}^{N}h_{T{\Theta}_n} h_{\Theta {R}_n} \mathrm{e}^{j\left(\theta_n-\alpha_n-\beta_{n}\right)}\right|^2}{{d}^\chi_{ST}{d}^\chi_{T\Theta}{d}^\chi_{\Theta R}{\sigma}^2_R}  \\ 
\overset{(a)}{=}&\bar{\gamma}_{R_1}|h_{ST}|^2|h_{TR}|^2+
\bar{\gamma}_{R_2}|h_{ST}|^2\left|\sum_{n=1}^{N}h_{T{\Theta}_n} h_{\Theta {R}_n}\right|^2 ,\label{gamma_Reader_2}
\end{align}
where $\bar{\gamma}_{R_1}$ and $\bar{\gamma}_{R_2}$ are the average SNR at the reader due to the direct and the RIS-aided links, respectively.
By re-writing \eqref{gamma_Reader_2} as $\gamma_R=\gamma_{R_1}+\gamma_{R_2}$, $f_{\gamma_{R_1}}(\gamma_{R_1})$ can be given by \cite{pdf_Fisher_2}
\begin{align}
\scalebox{0.99}{$\displaystyle f_{\gamma_{R_1}}(\gamma_{R_1}) = \frac{\eta_1}{\gamma_{R_1}} G_{2,2}^{2,2}\left( \hspace{-6pt} \begin{array}{c}
	\frac{\delta_1\gamma_{R_1}}{\bar{\gamma}_{R_1}}\end{array} \hspace{-5pt}
\Bigg\vert \hspace{-4pt} \begin{array}{c}
	1-m_{s_{ST}},1-m_{s_{TR}}\\
	m_{ST},m_{TR}\\
\end{array} \hspace{-6pt} \right) $} ,\label{pdf_GR_1}
\end{align}
where $\eta_1=\frac{1}{\Gamma(m_{ST}) \Gamma(m_{S_{ST}}) \Gamma(m_{TR}) \Gamma(m_{S_{TR}})}$ and $\delta_1= \frac{m_{ST} m_{TR}}{(m_{S_{ST}}-1) (m_{S_{TR}}-1)}$. According to Thm. \ref{theor_pdf_R}, $f_{\gamma_{R_2}}(\gamma_{R_2})$ can be obtained as
\begin{align} \label{f_{GR2}(GR2)}
 \scalebox{0.95}{$\displaystyle f_{\gamma_{R_2}}(\gamma_{R_2})=\mathcal{G}G^{1,3}_{3,1}\left( \hspace{-6pt} \begin{array}{c}	4\lambda_{1} {\bar{y}_1}^{-\frac{2}{d}} \gamma_{R_2} \end{array} \hspace{-2pt} 		\Big\vert \hspace{-2pt} \begin{array}{c} \frac{2-c}{2}, \frac{3-c}{2}, 2-m_{ST} \\ m_{S_{ST}}+1\\ \end{array} \hspace{-6pt} \right) $}.
 \end{align}
Now, since $\gamma_R=\gamma_{R_1}+\gamma_{R_2}$, we exploit the Moment-Generating function (MGF) of \begin{math}\gamma_{R_1}\end{math} and \begin{math}\gamma_{R_2}\end{math} to obtain the PDF and CDF of \begin{math}\gamma_{R}\end{math} as 
\begin{align}  \label{pdf-laplace}
 f_{\gamma_{R}}(\gamma_{R})= \mathscr{L}^{-1}\left\{M_{\gamma_{R_1}}(s)\ M_{\gamma_{R_2}}(s)\right\},
\end{align}
\begin{align}  \label{cdf-laplace}
 F_{\gamma_{R}}(\gamma_{R})= \mathscr{L}^{-1}\left\{\frac{1}{s}\ M_{\gamma_{R_1}}(s)\ M_{\gamma_{R_2}}(s)\right\},
\end{align}
where \begin{math}\mathscr{L}^{-1}\end{math} shows the Laplace inverse transform and \begin{math}M_\gamma(t)=M_\gamma(-s)\end{math} denotes the MGF of \begin{math}\gamma\end{math}. 

\vspace{0pt}
\begin{theorem} \label{theor_R_2}
Assuming all channels follow Fisher-Snedecor $\mathcal{F}$ fading distribution, the PDF and CDF of $\gamma_R$ with direct links can be obtained as \eqref{pdf_SNR_Reader} and \eqref{cdf_SNR_Reader}, respectively. 
\end{theorem}

\begin{proof}
    The proof is elaborated in Appendix C.
\end{proof}

\begin{figure*}[t]
			\normalsize
	\setcounter{equation}{24}
   \begin{align}  \label{pdf_SNR_Reader}
\scalebox{0.915}{$\displaystyle
f_{\gamma_{R}}(\gamma_{R}) \hspace{-2pt} = \hspace{-2pt} \mathcal{G} \hspace{1pt} \eta_1   {\gamma_{R}}^{-2}  H^{0,0: 2,3; 1,4}_{1,0: 3,2; 4,1} \hspace{-3pt} \left( \hspace{-8pt} \begin{array}{c}
\frac{\delta_1}{\bar{\gamma}_{R_1} \gamma_{R}} \\ 
\frac{4\lambda_1 {\bar{y}_1}^{-\frac{2}{d}}}{ \gamma_{R}}
    \end{array} \hspace{-6pt}
\Bigg\vert \hspace{-6pt} \begin{array}{c}
(-1;1,1):(1,1), (1-m_{S_{ST}}, 1), (1-m_{S_{TR}}, 1); (0,1), 
    (\frac{2-c}{2},1), (\frac{3-c}{2},1), (2-m_{ST}, 1)\\
			––––:
   (m_{ST},1), (m_{TR},1);
   (1+m_{S_{ST}}, 1)\\
\end{array} \hspace{-8pt} \right)$}.
			\end{align}
			\hrulefill
      \vspace{-5pt}

		\end{figure*}
  \setlength{\intextsep}{1pt plus 2pt minus 2pt}
\begin{figure*}[t]
			\normalsize
\setcounter{equation}{25}
\begin{align}  \label{cdf_SNR_Reader}
\scalebox{0.915}{$\displaystyle
F_{\gamma_{R}}(\gamma_{R}) \hspace{-2pt} = \hspace{-2pt} \mathcal{G} \hspace{1pt} \eta_1   {\gamma_{R}}^{-3}  H^{0,0: 2,3; 1,4}_{1,0: 3,2; 4,1} \hspace{-3pt} \left( \hspace{-8pt} \begin{array}{c}
\frac{\delta_1}{\bar{\gamma}_{R_1} \gamma_{R}} \\ 
\frac{4\lambda_1 {\bar{y}_1}^{-\frac{2}{d}}}{ \gamma_{R}}
    \end{array} \hspace{-6pt}
\Bigg\vert \hspace{-6pt} \begin{array}{c}
(-2;1,1):(1,1), (1-m_{S_{ST}}, 1), (1-m_{S_{TR}}, 1); (0,1), 
    (\frac{2-c}{2},1), (\frac{3-c}{2},1), (2-m_{ST}, 1)\\
			––––:
   (m_{ST},1), (m_{TR},1);
   (1+m_{S_{ST}}, 1)\\
\end{array} \hspace{-8pt} \right)$}.
			\end{align}
			\hrulefill
	\vspace{-5pt}		
		\end{figure*}

\subsubsection{Eavesdropper link}
According to \eqref{Eve_SNR}, the instantaneous SNR at Eve can be determined as
\begin{align} 
\gamma_E&=\frac{{\left|\sqrt{P_s}(h_{ST}h_{TE}+h_{ST}\vec{H}_{T\Theta}\vec{\Theta} \vec{H}_{\Theta E})\right|}^2}{n_k} \nonumber \\ 
 =&\scalebox{0.91}{$\displaystyle\frac{P_s|h_{ST}|^2|h_{TE}|^2}{{d}^\chi_{ST}{d}^\chi_{TE}{\sigma}^2_R} \hspace{-2pt} + \hspace{-2pt} \frac{P_s|h_{ST}|^2 \hspace{-2pt} \left|\sum_{n=1}^{N} \hspace{-2pt} h_{T{\Theta}_n} \hspace{-2pt} h_{\Theta {E}_n} \hspace{-2pt} \mathrm{e}^{j \hspace{-1pt} \left( \hspace{-1pt} \theta_n \hspace{-1pt} -\alpha_n \hspace{-1pt} -\epsilon_{n} \hspace{-1pt} \right)}\hspace{-1pt} \right|^2}{{d}^\chi_{ST}{d}^\chi_{T\Theta}{d}^\chi_{\Theta R}{\sigma}^2_E} $} 
 \\ 
=& \scalebox{0.89}{$\displaystyle \bar{\gamma}_{E_1}|h_{ST}|^2|h_{TR}|^2 \hspace{-2pt}+ \hspace{-2pt}\bar{\gamma}_{E_2}|h_{ST}|^2 \hspace{-3pt} \left|\sum_{n=1}^{N} \hspace{-2pt} h_{T{\Theta}_n} \hspace{-2pt} h_{\Theta {R}_n} \mathrm{e}^{j \hspace{-1pt} \left( \hspace{-1pt} \theta_n \hspace{-1 pt} - \hspace{-1pt} \alpha_n \hspace{-1pt} - \hspace{-2pt} \epsilon_{n} \hspace{-1pt} \right)} \hspace{-2pt} \right|^2 $},  \label{gamma_Eve_2}
\end{align}
where $\bar{\gamma}_{E_1}$ and $\bar{\gamma}_{E_2}$ are the average SNR at Eve due to the direct and the RIS-aided links, respectively. 
By re-writing \eqref{gamma_Eve_2} as $\gamma_E=\gamma_{E_1}+\gamma_{E_2}$, $f_{\gamma_{E_1}}(\gamma_{E_1})$ can be shown as follows \cite{pdf_Fisher_2}.
\begin{align}
\scalebox{0.99}{$\displaystyle f_{\gamma_{E_1}}(\gamma_{E_1}) = \frac{\eta_2}{\gamma_{E_1}} G_{2,2}^{2,2}\left( \hspace{-6pt} \begin{array}{c}
\frac{\delta_2\gamma_{E_1}}{\bar{\gamma}_{E_1}}\end{array} \hspace{-5pt}
\Bigg\vert \hspace{-4pt} \begin{array}{c}
	1-m_{s_{ST}},1-m_{s_{TE}}\\
	m_{ST},m_{TE}\\
\end{array} \hspace{-6pt} \right) $},\label{pdf_GE_1}
\end{align}
where $\eta_2=\frac{1}{\Gamma(m_{ST}) \Gamma(m_{S_{ST}}) \Gamma(m_{TE}) \Gamma(m_{S_{TE}})}$ and $\delta_2= \frac{m_{ST} m_{TE}}{(m_{S_{ST}}-1) (m_{S_{TE}}-1)}$.
As mentioned before, since Eve is a passive eavesdropper and the knowledge about its CSI is imperfect, 
and the equivalent channel reflected by RIS exhibits similarities to Rayleigh fading \cite{PHi_Error_Javier,PHi_Error_1,SOP_Discrete_PHi}. Thus, we can use the same analysis provided in Thm. \ref{theor_pdf_E} to obtain $f_{\gamma_{E_2}}(\gamma_{E_2})$  as
\begin{align} \label{f_{GE2}(GE2)}
 f_{\gamma_{E_2}}(\gamma_{E_2})= \frac{\mathcal{C}}{a\bar{\gamma}_{E_2}} G^{1,2}_{2,1}\left( \hspace{-6pt} \begin{array}{c}	a\lambda_{1}  \gamma_{E_2} \end{array} \hspace{-2pt} 		\Big\vert \hspace{-2pt} \begin{array}{c} 1,  2-m_{ST} \\ 1+m_{S_{ST}}\\ \end{array} \hspace{-6pt} \right).
 \end{align}
Now, by using the MGF of \begin{math}\gamma_{E_1}\end{math} and \begin{math}\gamma_{E_2}\end{math} and considering \eqref{pdf-laplace} and \eqref{cdf-laplace}, we can derive $f_{\gamma_{E}}(\gamma_{E})$ as the following theorem. 

\begin{theorem} \label{theor_E_2}
Assuming all channels follow the Fisher-Snedecor $\mathcal{F}$ fading distribution, the PDF and CDF of $\gamma_E$ with direct links can be obtained as \eqref{pdf_SNR_Eve_2} and \eqref{cdf_SNR_Eve_2}, respectively. 
\end{theorem}

\vspace{-0pt}
\begin{proof} The proof is elaborated in Appendix D.
\end{proof}

\begin{figure*}[t]
			\normalsize
	\setcounter{equation}{30}
   \begin{align}  \label{pdf_SNR_Eve_2}
f_{\gamma_{E}}(\gamma_{E}) \hspace{-2pt} = \hspace{-2pt}  \frac{\eta_2 \hspace{1 pt} \mathcal{C}}{\bar{\gamma}_{E_2}a}   {\gamma_{E}}^{-2}  H^{0,0: 2,3; 1,3}_{1,0: 3,2; 3,1} \hspace{-3pt} \left( \hspace{-8pt} \begin{array}{c}
\frac{\delta_2}{\bar{\gamma}_{E_1} \gamma_{E}} \\ 
\frac{a\lambda_1 }{ \gamma_{E}}
    \end{array} \hspace{-6pt}
\Bigg\vert \hspace{-6pt} \begin{array}{c}
(-1;1,1):(1,1), (1-m_{S_{ST}}, 1), (1-m_{S_{TE}}, 1); (2,1), 
    (1,1),  (2-m_{ST}, 1)\\
			––––:
   (m_{ST},1), (m_{TE},1);
   (1+m_{S_{ST}}, 1)\\
\end{array} \hspace{-8pt} \right).
			\end{align}
			\hrulefill
\vspace{-13pt}		\end{figure*}
  \setlength{\intextsep}{1pt plus 2pt minus 2pt}
\begin{figure*}[t]
			\normalsize
	\setcounter{equation}{31}
   \begin{align}  \label{cdf_SNR_Eve_2}
F_{\gamma_{E}}(\gamma_{E}) \hspace{-2pt} = \hspace{-2pt}  \frac{\eta_2 \hspace{1 pt} \mathcal{C}}{\bar{\gamma}_{E_2}a}   {\gamma_{E}}^{-3}  H^{0,0: 2,3; 1,3}_{1,0: 3,2; 3,1} \hspace{-3pt} \left( \hspace{-8pt} \begin{array}{c}
\frac{\delta_2}{\bar{\gamma}_{E_1} \gamma_{E}} \\ 
\frac{a\lambda_1 }{ \gamma_{E}}
    \end{array} \hspace{-6pt}
\Bigg\vert \hspace{-6pt} \begin{array}{c}
(-2;1,1):(1,1), (1-m_{S_{ST}}, 1), (1-m_{S_{TE}}, 1); (2,1), 
    (1,1),  (2-m_{ST}, 1)\\
			––––:
   (m_{ST},1), (m_{TE},1);
   (1+m_{S_{ST}}, 1)\\
\end{array} \hspace{-8pt} \right).
			\end{align}
			\hrulefill	
\vspace{-13pt}		\end{figure*}


\vspace{-0pt}
\section{Secrecy Performance Analysis}\label{sop_section}
In this section, we derive the analytical expressions of ASC and SOP for RIS-aided BC, exploiting the distributions obtained in the previous section.
 Secrecy capacity (SC) refers to the highest possible transmission rate at which information can be sent over the BC channels while ensuring that the transmitted information remains confidential. Thus, SC can be expressed as
\begin{align} \label{SC1}
C_s(\gamma_R, \gamma_E)= {\Big[C_R-C_E\Big]}^+,
\end{align}
where $C_R=\log_2\left(1+\gamma_R\right)$ and $C_E=\log_2\left(1+\gamma_E\right)$ denote the wireless channel capacity between the tag and reader, and the tag and Eve, respectively. 

\subsection{Without Direct Links}
\subsubsection{ASC Analysis}
Since the SNRs are random variables, and thus, the SC is a random variable, we need to apply an expectation over the achievable secrecy rate to evaluate the system performance. 
ASC represents the average value of the SC across various potential channel conditions and serves as a crucial metric in assessing the PLS performance.
Referring to \eqref{SC1}, when considering a complex AWGN wiretap channel, the SC is defined as the difference between the main channel and the eavesdropper channel, specifically when Eve's channel experiences more noise than the main channel. By utilizing this definition and derived PDFs and CDFs in Thms. \ref{theor_pdf_R} and \ref{theor_pdf_E}, we can express the mathematical expression of the ASC as
\begin{align} \label{ASC_1}
\bar{C}_s\overset{\Delta}{=} \hspace{-4 pt} \int_{0}^{\infty} \hspace{-3 pt} \int_{0}^{\infty} \hspace{-2 pt} C_s(\gamma_R,\gamma_E) f_{\gamma_R}(\gamma_R)f_{\gamma_E}(\gamma_E)  \mathrm{d}{\gamma_R} \mathrm{d}{\gamma_E}.
\end{align}
Therefore, the ASC for the considered RIS-aided BC is derived in the following theorem.
\begin{theorem}  \label{theor_ASC_1}
The ASC for the considered RIS-aided BC system without direct links under Fisher-Snedecor $\mathcal{F}$ fading channels is given by \eqref{ASC1}, where $\mathcal{A}=\frac{\mathcal{G}\mathcal{C}}{a\bar{\gamma}_Eln(2)}$.
\end{theorem}
\begin{figure*}[t]
			\normalsize
	\setcounter{equation}{34}
   \begin{align}  \label{ASC1}
\bar{C}_s \hspace{0pt} &= \hspace{0pt}  \mathcal{A}  H^{2,1: 1,3; 1,3}_{2,2: 3,2; 3,2} \hspace{-1pt} \left( \hspace{-4pt} \begin{array}{c} \vspace{3pt}
a\lambda_1 \\ 
\frac{4\lambda_1 }{{\bar{y}_1}^\frac{2}{d}}
    \end{array} \hspace{-3pt}
\Bigg\vert \hspace{-4pt} \begin{array}{c}
(-2;1,1),(-1;1,1) :(0,1), (1,1),(2-m_{ST}, 1); (\frac{2-c}{2},1),(\frac{3-c}{2},1),(2-m_{ST}, 1)  \\
			(-2;1,1),(-2;1,1):
   (1+m_{S_{ST}}, 1), (-1,1);
   (1+m_{S_{ST}}, 1)\\
\end{array} \hspace{-5pt} \right)  \nonumber
\\  &+   
\hspace{0pt}  \mathcal{A}  H^{2,1: 1,4; 1,2}_{2,2: 4,2; 2,1} \hspace{-1pt} \left( \hspace{-4pt} \begin{array}{c} \vspace{3pt}
\frac{4\lambda_1 }{{\bar{y}_1}^\frac{2}{d}} \\ a\lambda_1 
    \end{array} \hspace{-3pt}
\Bigg\vert \hspace{-4pt} \begin{array}{c}
(-2;1,1),(-1;1,1) :(0,1), (\frac{2-c}{2},1),(\frac{3-c}{2},1),(2-m_{ST}, 1) ; (1,1),(2-m_{ST}, 1) \\
			(-2;1,1),(-2;1,1):
   (1+m_{S_{ST}}, 1), (-1,1);
   (1+m_{S_{ST}}, 1)\\
\end{array} \hspace{-5pt} \right)  \nonumber 
\\  &-  
\frac{\mathcal{A}}{\mathcal{G}}G^{3,3}_{4,3}\left( \hspace{-6pt} \begin{array}{c}	a\lambda_{1}  \end{array} \hspace{-2pt} 		\Big\vert \hspace{-2pt} \begin{array}{c} 1,  2-m_{ST}, -1, 0 \\ 1+m_{S_{ST}}, -1, -1 \\ \end{array} \hspace{-6pt} \right).
			\end{align}
			\hrulefill
\vspace{-13pt}
		\end{figure*}

\begin{proof}
The proof is elaborated in Appendix E.
\end{proof}

\subsubsection{SOP Analysis}
SOP is an important analytical measure used to assess the performance of PLS, which quantifies the probability that SC falls below a specific positive secrecy rate threshold, say $R_s>0$, i.e.,
\begin{align}
    P_{\mathrm{sop}}=\Pr\left(C_s\le R_s\right).
\end{align}
Now, by inserting \eqref{SC1} into SOP definition we have
\begin{align} 
P_\mathrm{sop}&=\Pr\left(\ln\left(\frac{1+\gamma_R}{1+\gamma_E}\right)\le R_s\right) 
\\ \label{SOP2}
&=\int_{0}^{\infty}F_{\gamma_R}(\gamma_t)f_{\gamma_E}(\gamma_E) d_{\gamma_E},
\end{align}
in which $\gamma_t=(1+\gamma_E) \mathrm{e}^{R_s}-1=\gamma_E \mathrm{e}^{R_s}+\mathrm{e}^{R_s}-1=\gamma_E R_t+{R'}_t$ is the SNR threshold.

\begin{theorem} \label{theor_SOP_1}
The SOP for the considered RIS-aided BC system without direct links under Fisher-Snedecor $\mathcal{F}$ fading channels is given by \eqref{SOP1}.
\end{theorem}
\begin{figure*}[t]
			\normalsize
	\setcounter{equation}{38}
   \begin{align}  \label{SOP1}
&P_\mathrm{sop}\hspace{0pt}=  \hspace{0pt}  \frac{\mathcal{G} R'^2_t\mathcal{C}}{a\bar{\gamma}_{E_2}R_t}     H^{0,1: 4,1; 3,1}_{1,0: 2,5; 1,3} \hspace{-2pt} \left( \hspace{-6pt} \begin{array}{c}
\frac{{\bar{y}_1}^\frac{2}{d}}{4\lambda_1 R'_t} \\ 
\frac{R_t}{ a\lambda_1 R'_t}
    \end{array} \hspace{-4pt}
\Bigg\vert \hspace{-4pt} \begin{array}{c}
(3;1,1):(1,1), (-m_{S_{ST}}, 1), (2,1); (-m_{S_{ST}}, 1) \\
			––––: (1,1),(\frac{c}{2},1), (\frac{c-1}{2},1),
   (m_{ST},1), (2,1); (0,1), 
   (m_{ST}-1,1), (1,1)\\
\end{array} \hspace{-4pt} \right).
			\end{align}
			\hrulefill
\vspace{-5pt}
		\end{figure*}
\begin{proof} The proof is elaborated in Appendix F.
\end{proof}

\subsection{With Direct Links}

\subsubsection{ASC Analysis} By utilizing the ASC definition in \eqref{SC1} and the obtained PDFs and CDFs in Thms. \ref{theor_R_2} and \ref{theor_E_2}, 
the ASC for RIS-aided BC with considering the direct links can be derived as the following theorem. 

\begin{theorem} \label{theor_ASC2}
    The ASC for RIS-aided BC with direct links under Fisher-Snedecor $\mathcal{F}$ fading channels is given by \eqref{ASC02}, where $\mathcal{A'}=\frac{\eta_1\eta_2\mathcal{G}\mathcal{C}}{a\bar{\gamma}_{E_2}ln(2)}$,   $\lambda_1= (-4;1,1,1,1),(-3;1,1,1,1)$, $\lambda_2= (1,1), (1-m_{S_{ST}}, 1), (1-m_{S_{TR},} 1)$, $\lambda_3= (0,1), (\frac{2-c}{2}, 1), (\frac{3-c}{2}, 1),   (2-m_{ST}, 1)$, $\lambda_4= (0,1), (1-m_{S_{ST}, 1}), (1-m_{S_{TE}, 1})$, $\lambda_5= (0,1), (1,1), (2-m_{ST}, 1)$, $\Pi_1= (-3;1,1,1,1),(-3;1,1,1,1)$, $\Pi_2= (m_{ST},1),(m_{TR},1)$, $\Pi_3= (1+m_{S_{ST}} , 1)$, $\Pi_4= (m_{ST},1), (m_{TE},1)$, $\Pi_5= (1+m_{S_{ST}} , 1)$, $\lambda'_1= (-3;1,1,1,1),(-3;1,1,1,1),(-3;1,1,1,1)$, $\Pi'_1= (-4;1,1,1,1)$, $\mathcal{A''}=\frac{\eta_2\mathcal{C}}{a\bar{\gamma}_{E_2}ln(2)}$, $\tau_1=(-1;1,1), (-1;1,1),(0;1,1)$, $\tau_2=(1,1), (1-m_{S_{ST}}, 1), (1-m_{S_{TE}}, 1)$, $ \tau_3=(2,1), (1,1),  (2-m_{ST}, 1)$, $\nu_1=(0;1,1),(0;1,1)$, $\nu_2=(m_{ST}, 1), (m_{TE}, 1)$, and $\nu_3= (1+m_{S_{ST}}, 1)$.
\end{theorem}
\begin{figure*}[t]
			\normalsize
			\setcounter{equation}{39}
   \begin{align}  \label{ASC02}
\bar{C}_s&= \mathcal{A'} \hspace{1pt}  H^{0,1: 2,3; 1,4; 2,3; 1,3}_{2,2: 3,2; 4,1; 3,2; 3,1}  \left(\hspace{-4pt} \begin {array}{c} 
	\frac{\delta_1}{\bar{\gamma}_{R_1}}, 4\lambda_1{\bar{y}_1}^{\frac{-2}{d}}, \frac{\delta_2}{\bar{\gamma}_{E_1}},
             a\lambda_1
    \end{array} \hspace{-2pt}
			\Bigg \vert \hspace{-2pt} \begin{array}{c}
	\Lambda_1: \Lambda_2; \Lambda_3; \Lambda_4; \Lambda_5
    \\
			\Pi_1:\Pi_2; \Pi_3;  \Pi_4;  \Pi_5
   \\
\end{array} \hspace{-3pt}\right)  \nonumber \\ &+    \mathcal{A'} \hspace{1pt}  H^{0,2: 2,3; 1,3; 2,3; 1,4}_{3,1: 3,2; 3,1; 3,2; 4,1}  \left(\hspace{-4pt} \begin {array}{c} 
		\frac{\delta_2}{\bar{\gamma}_{E_1}},
    a\lambda_1,   
    \frac{\delta_1}{\bar{\gamma}_{R_1}}, 4\lambda_1{\bar{y}_1}^{\frac{-2}{d}}
    \end{array} \hspace{-2pt}
			\Bigg \vert \hspace{-2pt} \begin{array}{c}
	\Lambda'_1: \Lambda_4; \Lambda_5; \Lambda_2; \Lambda_3
    \\
			\Pi'_1:\Pi_4; \Pi_5;  \Pi_2;  \Pi_3
   \\
\end{array} \hspace{-3pt}\right)  
 \hspace{-2pt}  +    \mathcal{A''} \hspace{1pt}  H^{0,1: 2,3; 1,3}_{3,2: 3,2; 3,1} \hspace{-2pt} \left(\hspace{-6pt} \begin {array}{c} 
		\frac{\delta_2}{\bar{\gamma}_{E_1}},
    a\lambda_1 
    \end{array} \hspace{-4pt}
			\Bigg \vert \hspace{-3pt} \begin{array}{c}
	\tau_1: \tau_2; \tau_3
    \\
			\nu_1:\nu_2; \nu_3
   \\
			\end{array} \hspace{-5pt}\right).
			\end{align}
\vspace{-10pt}			\hrulefill
		\end{figure*}
\begin{proof}
    The proof is elaborated in Appendix G.
\end{proof}

\subsubsection{SOP Analysis}  By utilizing the SOP definition in \eqref{SOP2} and the obtained PDFs and CDFs in Thms. \ref{theor_R_2} and \ref{theor_E_2}, 
the SOP for the considered system model can be obtained as the following theorem. 
\begin{theorem} \label{theor_SOP2}
    The SOP for the considered RIS-aided BC system with direct links under Fisher-Snedecor $\mathcal{F}$ fading channels is given by \eqref{SOP02}, where $\mathcal{P}=\frac{\eta_1\eta_2\mathcal{G}\mathcal{C}R_t}{a\bar{\gamma}_{E_2}R^{'4}_t}$, $\Xi_1=(1,1), (1-m_{S_{ST}, 1}), (1-m_{S_{TR}, 1})$, $\Xi_2= (0,1), (\frac{2-c}{2}, 1), (\frac{3-c}{2}, 1),   (2-m_{ST}, 1)$, $\Xi_3= (1,1), (1-m_{S_{ST}, 1}), (1-m_{S_{TE}, 1})$, $\Xi_4= (0,1), (1,1), (2-m_{ST}, 1), (1+m_{S_{ST}, 1})$, $\Upsilon_1= (-3;1,1,1,1)$, $\Upsilon_2= (m_{ST},1),(m_{TR},1)$, $\Upsilon_3= (1+m_{S_{ST}, 1})$, $\Upsilon_4= (m_{ST},1), (m_{TE},1)$, and $\Upsilon_5= (1+m_{S_{ST}, 1})$.
\end{theorem}
\begin{figure*}[t]
			\normalsize
			\setcounter{equation}{40}
   \begin{align}  \label{SOP02}
&P_\mathrm{sop} =  \mathcal{P} \hspace{1pt}  H^{0,0: 2,3; 1,4; 2,3; 1,3}_{0,1: 3,2; 4,1; 3,2; 3,1}  \left(\hspace{-4pt} \begin {array}{c} 
		\frac{\delta_1}{\bar{\gamma}_{R_1}R'_t}, 
        \frac{4\lambda_1}{{\bar{y}_1}^{\frac{2}{d}}R'_t}, 
            \frac{\delta_2R_t}{\bar{\gamma}_{E_1}},
             a\lambda_1R_t
    \end{array} \hspace{-2pt}
			\Bigg \vert \hspace{-2pt} \begin{array}{c}
	––: \Xi_1; \Xi_2; \Xi_3; \Xi_4
    \\
			\Upsilon_1:\Upsilon_2; \Upsilon_3;  \Upsilon_4;  \Upsilon_5
   \\
			\end{array} \hspace{-3pt}\right).
			\end{align}
			\hrulefill
\vspace{-10pt}		\end{figure*}
\begin{proof}
    The proof is elaborated in Appendix H.
\end{proof}
\newtheorem{remark}{Remark}
\begin{remark}
The ASC and SOP behavior of RIS-aided BC systems under Fisher-Snedecor $\mathcal{F}$ fading channels are accurately modeled as shown in Thms. \ref{theor_ASC_1}–\ref{theor_SOP2}. 
The intricacies of these theorems reveal how different parameters intricately influence the system's secrecy performance.
Thms. \ref{theor_ASC_1} and \ref{theor_SOP_1} demonstrate that the ASC and SOP in the absence of direct links are primarily governed by the parameters within the Fox's H-function expressions, offering a nuanced understanding of how they modulate secrecy performance. These parameters include the number of RIS reflecting elements, the distance of every entity from the RIS, average SNR at Eve (\(\bar{\gamma}_{E_2}\)), and the fading and shadowing parameters (\(m_{ST}\), and \(m_{S_{ST}}\)), which dictate system susceptibility to channel variability. A higher average SNR at Eve or more severe fading conditions tend to degrade the secrecy performance, while a larger number $N$ or more favorable links at the reader lead to reduced SOP and increased ASC, as evidenced by the derived expressions.
The inclusion of direct links in Thms. \ref{theor_ASC2} and \ref{theor_SOP2} adds another layer of complexity. While the direct links provide an additional pathway for signal transmission, our analysis indicates that their contribution to enhancing secrecy performance is relatively minor compared to the significant role played by the RIS. This is particularly evident in the Fox's H-function expressions of Thms. \ref{theor_ASC2} and \ref{theor_SOP2}, where the direct link parameters subtly alter the overall secrecy metrics. The mathematical transition from Thms. \ref{theor_ASC_1} and \ref{theor_SOP_1} to Thms. \ref{theor_ASC2} and \ref{theor_SOP2}, though marked by an increase in complexity, underscores the dominant influence of RIS in shaping the secrecy performance of BC systems. The marginal improvement in secrecy performance with direct links highlights the criticality of optimizing the RIS configuration for achieving robust secrecy in BC systems.
\end{remark}

\section{Asymptotic Analysis of Secrecy Metrics} \label{sec_asy}

In this section, given the importance of the secrecy metrics performance in the high SNR regime, we evaluate the asymptotic
behaviour of both SOP and ASC by
exploiting the residue approach \cite{residal1}.

\subsection{Asymptotic ASC}
Since the exact analytical expression of ASC in first and second trms of \eqref{ASC1} is in terms of the bivariate Fox' H-function, we can derive the asymptotic behavior of the ASC at the high SNR regime (i.e., $\bar{\gamma}_R\rightarrow \infty$) by using the expansion of the bivariate Fox's H-function. To do this, we need to evaluate the residue of the corresponding integrands at the closest poles to the contour, namely, the minimum pole on the right for large Fox's H-function arguments and the maximum pole on the left for small ones. Hence, the asymptotic ASC can be determined according to the following proposition.

\begin{proposition} \label{asy_ASC_01}
The asymptotic ASC (i.e., $\bar{\gamma}_R\rightarrow \infty$) for the considered RIS-aided BC system under Fisher-Snedecor $\mathcal{F}$ fading channels is given by

\begin{align} \label{asy_ASC123}
\bar{C}_{\mathrm{s}}^{\mathrm{asy}} &= \mathcal{G}_1  G^{4,4}_{4,5}\left( \hspace{-6pt} \begin{array}{c}	\frac{{\bar{y}_1}^{\frac{2}{d}}a}{4}  \end{array} \hspace{-3pt} 		\Bigg\vert \hspace{-3pt} \begin{array}{c} 0, 1,2-m_{ST}, -2-m_{S_{ST}} \vspace{3pt} \\ 1+m_{S_{ST}}, \frac{c-4}{2}, \frac{c-5}{2}, -3+m_{ST}, -1 \\ \end{array} \hspace{-6pt} \right)   \nonumber \\ &+ 
\mathcal{G}_2  G^{3,5}_{5,4}\left( \hspace{-6pt} \begin{array}{c}	\frac{4}{{\bar{y}_1}^{\frac{2}{d}}a}  \end{array} \hspace{-3pt} 		\Bigg\vert \hspace{-3pt} \begin{array}{c} -2-m_{S_{ST}}, 0, \frac{2-c}{2}, \frac{3-c}{2}, 2-m_{ST}  \vspace{3pt} \\ 1+m_{S_{ST}}, -2,  -3+m_{ST}, -1 \\ \end{array} \hspace{-6pt} \right) \nonumber
\\ &-  \frac{\mathcal{A}}{\mathcal{G}}G^{3,3}_{4,3}\left( \hspace{-6pt} \begin{array}{c}	a\lambda_{1}  \end{array} \hspace{-2pt} 		\Bigg\vert \hspace{-2pt} \begin{array}{c} 1,  2-m_{ST}, -1, 0 \vspace{3pt} \\ 1+m_{S_{ST}}, -1, -1 \\ \end{array} \hspace{-6pt} \right),
\end{align}
where $\mathcal{G}_1 =\frac{\mathcal{G} \mathcal{C} \bar{y}_{1}^{\frac{4}{d}}}{16a\lambda_1^2\bar{\gamma}_{E}}$ and $\mathcal{G}_2 =\frac{\mathcal{G} \mathcal{C}} {a^3 \lambda_1^2\bar{\gamma}_{E}}$. 
\end{proposition}

\begin{proof}
    The proof is elaborated in Appendix I.
\end{proof}

\subsection{Asymptotic SOP}
With same strategy, we can derive the asymptotic behavior of the SOP at the high SNR regime (i.e., $\bar{\gamma}_R\rightarrow \infty$) by using the expansion of the bivariate Fox's H-function. 
Hence, the asymptotic SOP can be determined as follows.

\begin{proposition} \label{asy_SOP_01}
The asymptotic SOP (i.e., $\bar{\gamma}_R\rightarrow \infty$) for the considered RIS-aided BC system under Fisher-Snedecor $\mathcal{F}$ fading channels is given by
\begin{align}
\scalebox{0.93}{$\displaystyle P_{\mathrm{sop}}^{\mathrm{asy}}=\frac{\mathcal{G} \mathcal{C}R_t}{a^3\lambda_1^2\bar{\gamma}_{E_2}} G^{5,2}_{3,5}\left( \hspace{-6pt} \begin{array}{c}	\frac{{\bar{y}_1}^{\frac{2}{d}}a}{4R_t}  \end{array} \hspace{-2pt} 		\Big\vert \hspace{-2pt} \begin{array}{c} -m_{S_{ST}}, m_{ST}+4,2\\\frac{c}{2}, \frac{c-1}{2}, m_{ST}, 1, m_{S_{ST}}+3 \\ \end{array} \hspace{-6pt} \right) $}.
\end{align}
\end{proposition}

\begin{proof}
The proof is elaborated in Appendix J.
\end{proof}

\begin{remark}
The asymptotic analysis of ASC and SOP in high SNR regimes, as presented in Props. \ref{asy_ASC_01} and \ref{asy_SOP_01}, offers profound insights into the secrecy performance of RIS-aided BC systems under Fisher-Snedecor $\mathcal{F}$ fading channels. The utilization of the residue approach to expand the bivariate Fox's H-function is a testament to the depth of our analysis. This method allows us to accurately capture the essence of ASC and SOP behaviors as the system approaches high SNR limits, offering a deeper understanding of system performance under such conditions. The derived expressions highlight the delicate balance among various system parameters, such as the number of RIS elements, fading/shadowing characteristics, and SNR levels. This asymptotic perspective is crucial, as it not only validates the robustness of our system model in varying conditions but also provides valuable benchmarks for system design and optimization in practical deployment scenarios. The results from this analysis underscore the significant impact of RIS in enhancing ASC and SOP, particularly in high SNR regimes, thus reinforcing the pivotal role of RIS in designing secure BC systems.
\end{remark}

\section{Simulation Results}\label{num-results}

In this section, we validate the theoretical expressions of the derived ASC and SOP for RIS-aided BC through Monte-Carlo simulations. We conduct simulations for various BC scenarios, including cases with only a direct link, only RIS-aided links, and both direct and RIS-aided links. We also evaluate the PLS performance of RIS-aided BC based on different system parameters. 

\subsection{Simulation Setup}

We consider an RIS-aided BC system, featuring a tag with ultra-limited resources, a passive eavesdropper, a reader, and an RIS with $M$ reflecting elements. 
As shown in Fig. \ref{fig_simulation_setup}, the tag remains stationary at the coordinates (0,0,0), transmitting its confidential information to the reader by modulating and reflecting the RF signal emitted from source. 
\begin{figure}[t]
    \centering
\includegraphics[width=0.36\textwidth]{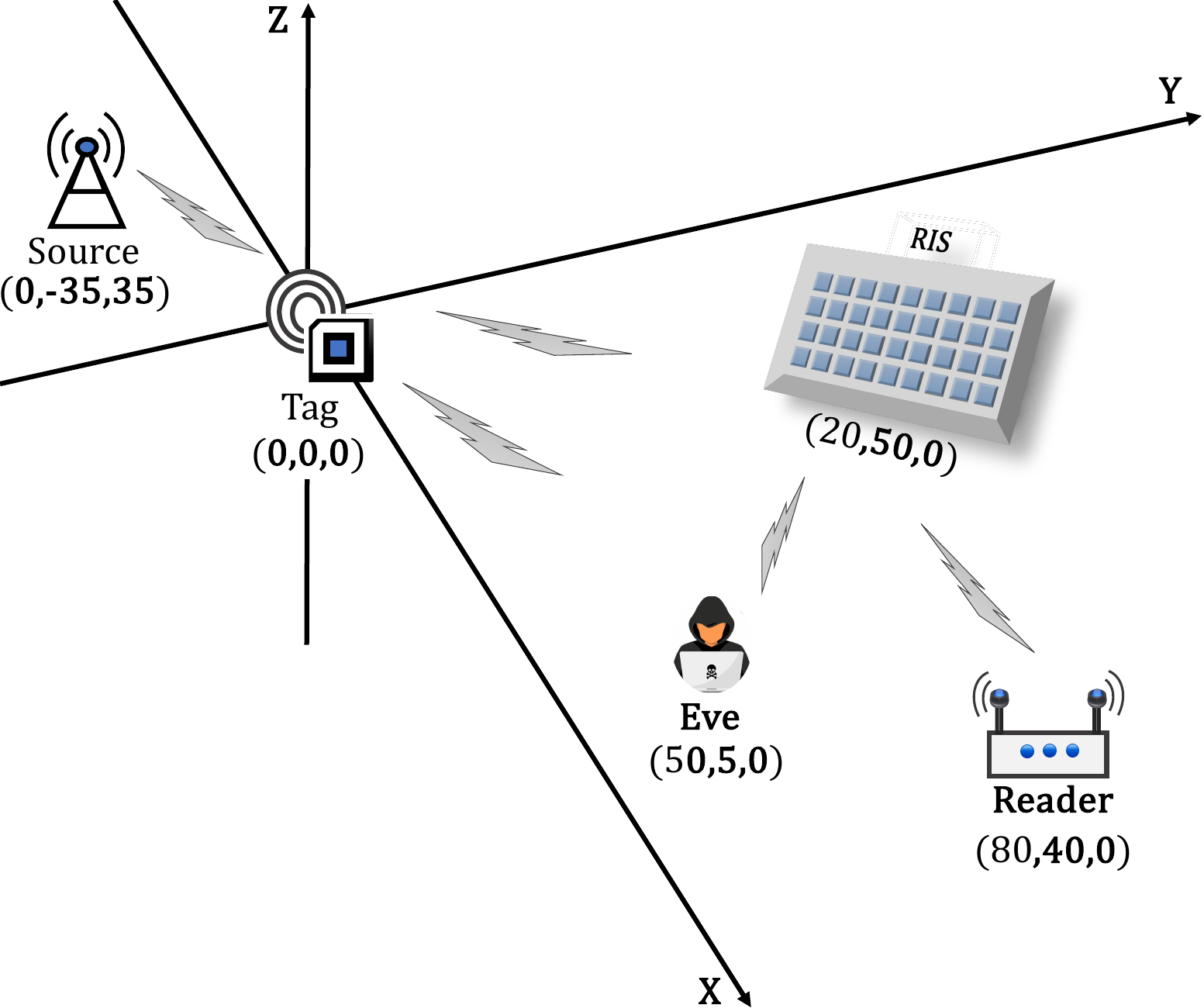}
    \caption{The simulation setup.}
    \label{fig_simulation_setup}
\end{figure}
The RIS is located at coordinates (20,50,0) to improve the received SNR at the reader, resulting in enhancing the BC system's PLS performance. Eve and the reader are located at (50,5,0) and (80,40,0), respectively, with the reader being positioned further away from the RIS and the tag compared to Eve. This setup, which represents a worst-case scenario, facilitates a detailed analysis of performance and results in the following fixed distances:
$d_{ST}= d_{T\Theta}= d_{TE}= d_{\Theta E}= 50$m, $d_{\Theta R}=60$m, and $d_{TR}=90$m, each derived from the framework of long-range BC \cite{ref_LoRa_BC1}. 
Additionally, the setup includes other parameters like $R_s=1$bps/Hz, $\sigma^2_R=-60$dbm, $\sigma^2_E=-40$dbm, $P_s=30$dBm, and $\chi =3.5$.
It is important to note that while the extended generalized bivariate/multivariate Fox's H-function is not readily available in popular mathematical software tools, one can implement it in MATLAB using programming functions presented in \cite{ref58}.


\subsection{Results and Discussions}

Fig. \ref{fig_ASC_01} illustrates the behavior of ASC under Fisher-Snedecor $\mathcal{F}$ fading channels for various values of $\bar{\gamma}_{R_2}$ and different numbers of RIS elements, denoted as $N$. As depicted in the figure, the ASC consistently increases with higher values of $\bar{\gamma}_{R_2}$ for a fixed $N$. Additionally, ASC shows an upward trend as $N$ increases, signifying that the utilization of RIS leads to a superior channel quality, enabling enhanced SNR at the reader through optimal phase adjustments.
\begin{figure}[t]
    \centering    \includegraphics[width=0.33\textwidth]{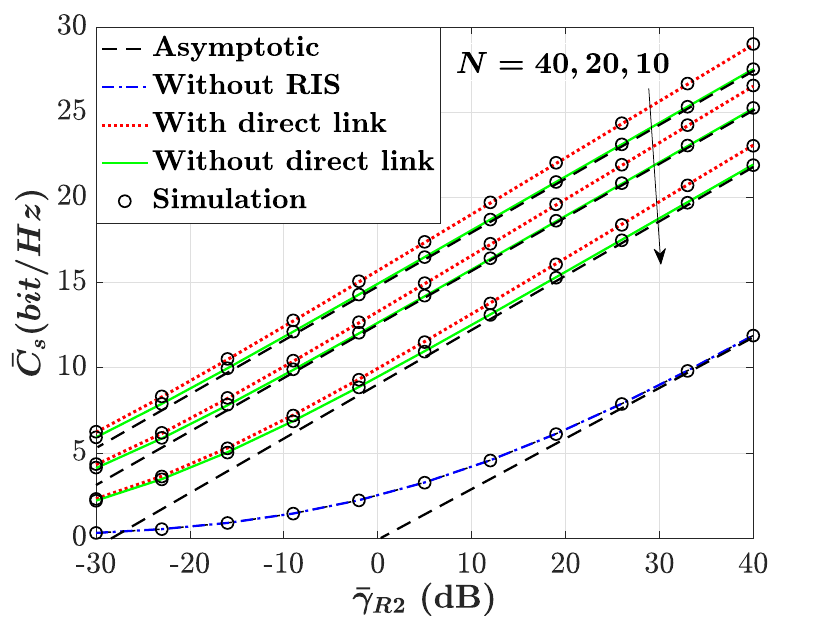}
    \caption{ASC versus $\Bar{\gamma}_{R_2}$ for different 
        numbers of RIS elements.}
    \label{fig_ASC_01}
\end{figure}
Fig. \ref{fig_ASC_02} demonstrates the variation of ASC concerning $\bar{\gamma}_{R_2}$ with fixed $N$ while considering selected values of $\bar{\gamma}_{E_2}$. Notably, ASC exhibits a consistent increase as $\bar{\gamma}_{R_2}$ rises, which is reasonable as it indicates improved main channel conditions. Additionally, as $\bar{\gamma}_{E_2}$ increases, ASC tends to decrease for a fixed $\bar{\gamma}_{R_2}$. Even when $\bar{\gamma}_{R_2} < \bar{\gamma}_{E_2}$, ASC remains greater than $4$ bits, implying that some level of secure communication is still achievable. However, the significance of secure communication becomes more pronounced with larger ASC values, which are attainable when the main channel's condition is superior to that of the eavesdropper ($\bar{\gamma}_{R_2} \geq \bar{\gamma}_{E_2}$) in practical BC scenarios.
\begin{figure}[t]
    \centering    \includegraphics[width=0.35\textwidth]{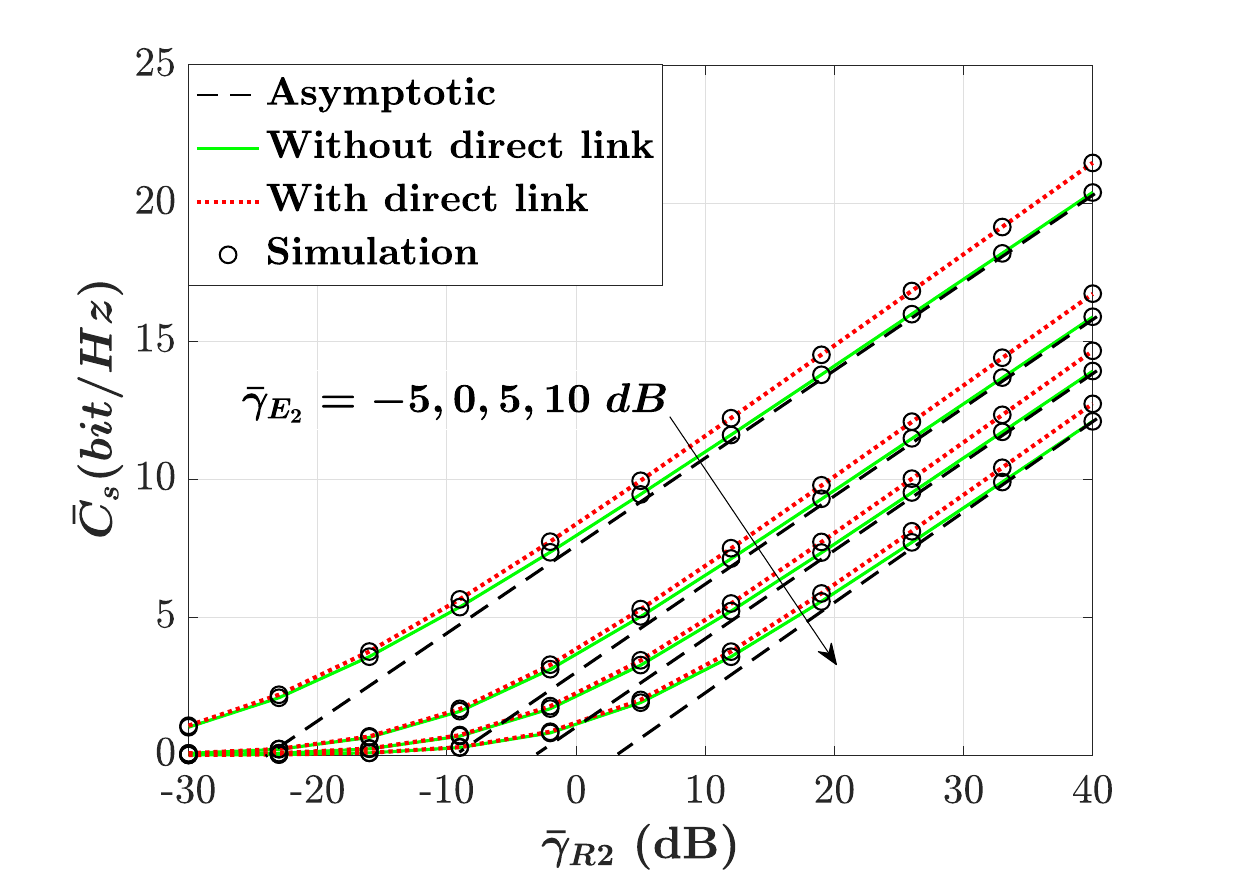}
    \caption{ASC versus $\Bar{\gamma}_{R_2}$ for different values of \\ $\Bar{\gamma}_{E_2}$ and $N=10$.}
    \label{fig_ASC_02}
\end{figure}
Fig. \ref{fig_ASC_03} illustrates the behavior of ASC concerning $\bar{\gamma}_{R_2}$ for various fading parameters $m_i$, $i \in \{ST, TR, TE\}$, in the presence of RIS deployment. Notably, as the fading becomes less severe (as $m_i$ increases), the ASC performance improves, indicating that the system's performance enhances (or degrades) in environments with lighter (or heavier) fading characteristics.
\begin{figure}[t]
    \centering    \includegraphics[width=0.35\textwidth]{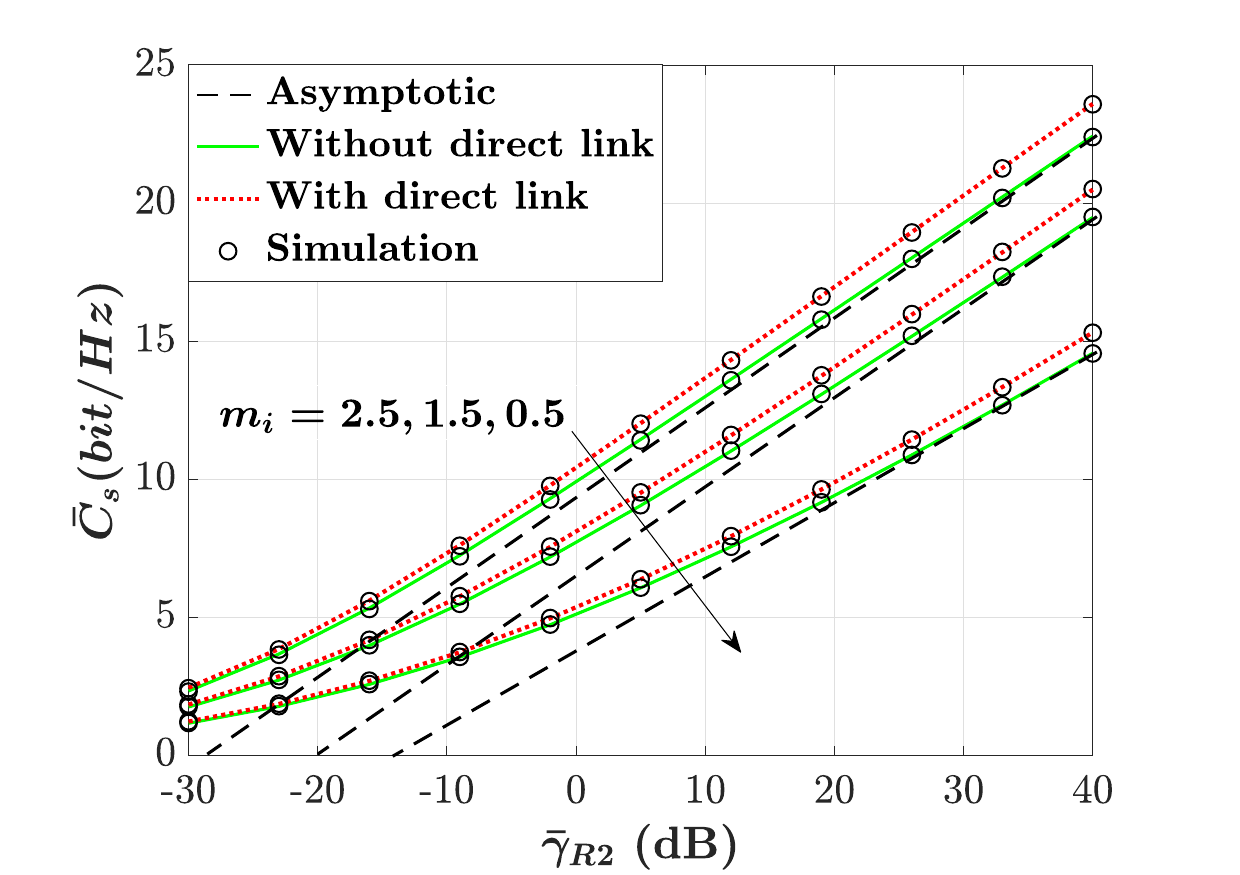}
    \caption{ASC versus ${S_R}$ for different values of $m_i$.}
    \label{fig_ASC_03}
\end{figure}

In practical BC systems, achieving secure communication is of paramount importance with the least possible value of SOP, as secure communication cannot be maintained when the system experiences a secrecy outage.
Fig. \ref{fig:SOP_01} demonstrates the variation of SOP concerning different values of $\bar{\gamma}_{R_2}$ and selected values of $N$. It is evident from the figure that SOP consistently decreases as $\bar{\gamma}_{R_2}$ increases for a fixed value of $N$. Additionally, as $N$ increases, SOP decreases since RIS deployment provides a higher-quality channel, leading to improved SNR at the reader. 
\begin{figure}[t]
    \centering    \includegraphics[width=0.35\textwidth]{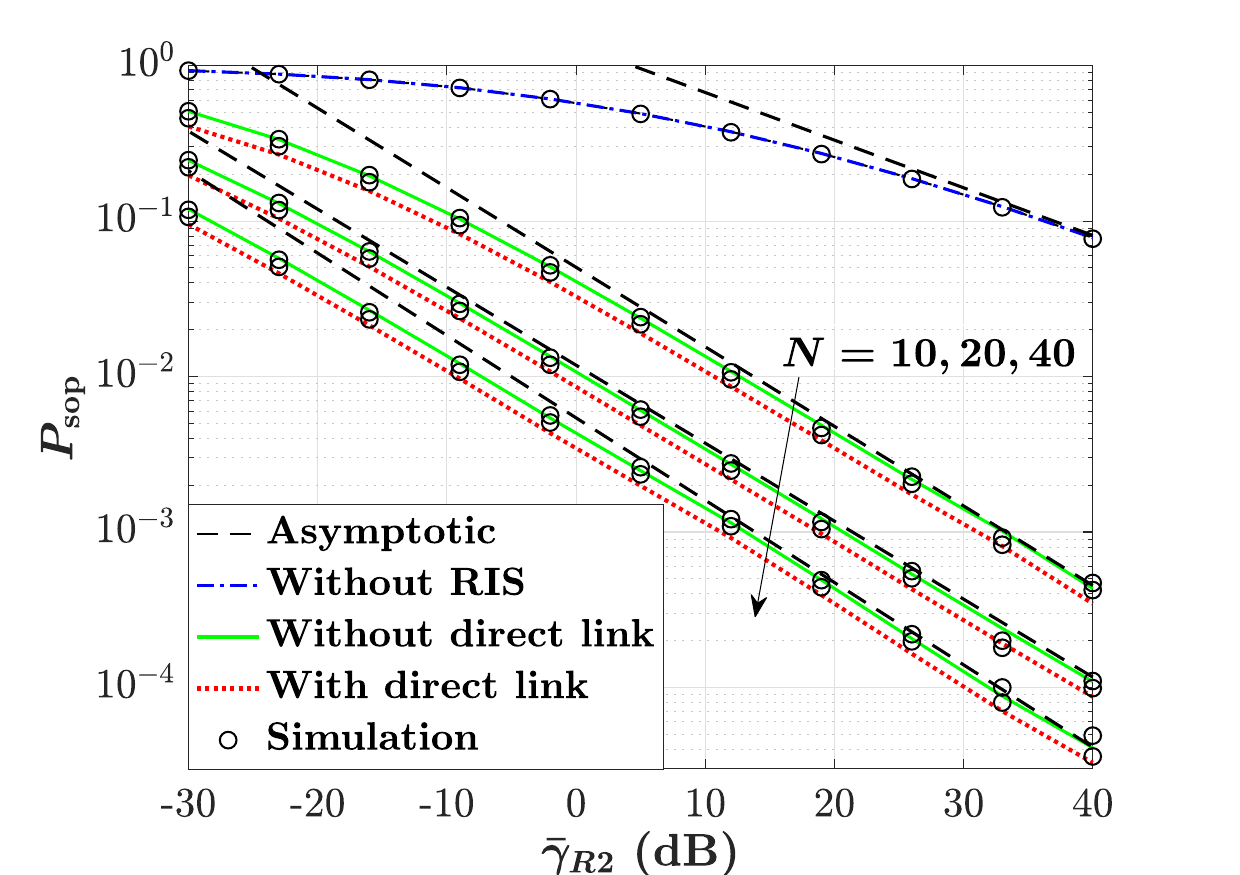}
    \caption{SOP versus $\Bar{\gamma}_{R_2}$ for different numbers of RIS elements.}
    \label{fig:SOP_01}
\end{figure}
Figure \ref{fig:SOP_02} presents a comprehensive insight into the behavior of SOP concerning $d_{\Theta R}$, while considering various values of $N$. 
As $d_{\Theta R}$ increases, the SOP exhibits a steady rise. This behavior emphasizes SOP becomes increasingly compromised when the RIS and the reader are positioned farther apart in the RIS-aided BC system. The higher SOP associated with larger distances underscores the importance of managing physical proximity to maintain robust secrecy in the communication process.
In addition, it can be observed that as $N$ grows, the figure highlights a consistent reduction in SOP. 
\begin{figure}[t]
    \centering    \includegraphics[width=0.35\textwidth]{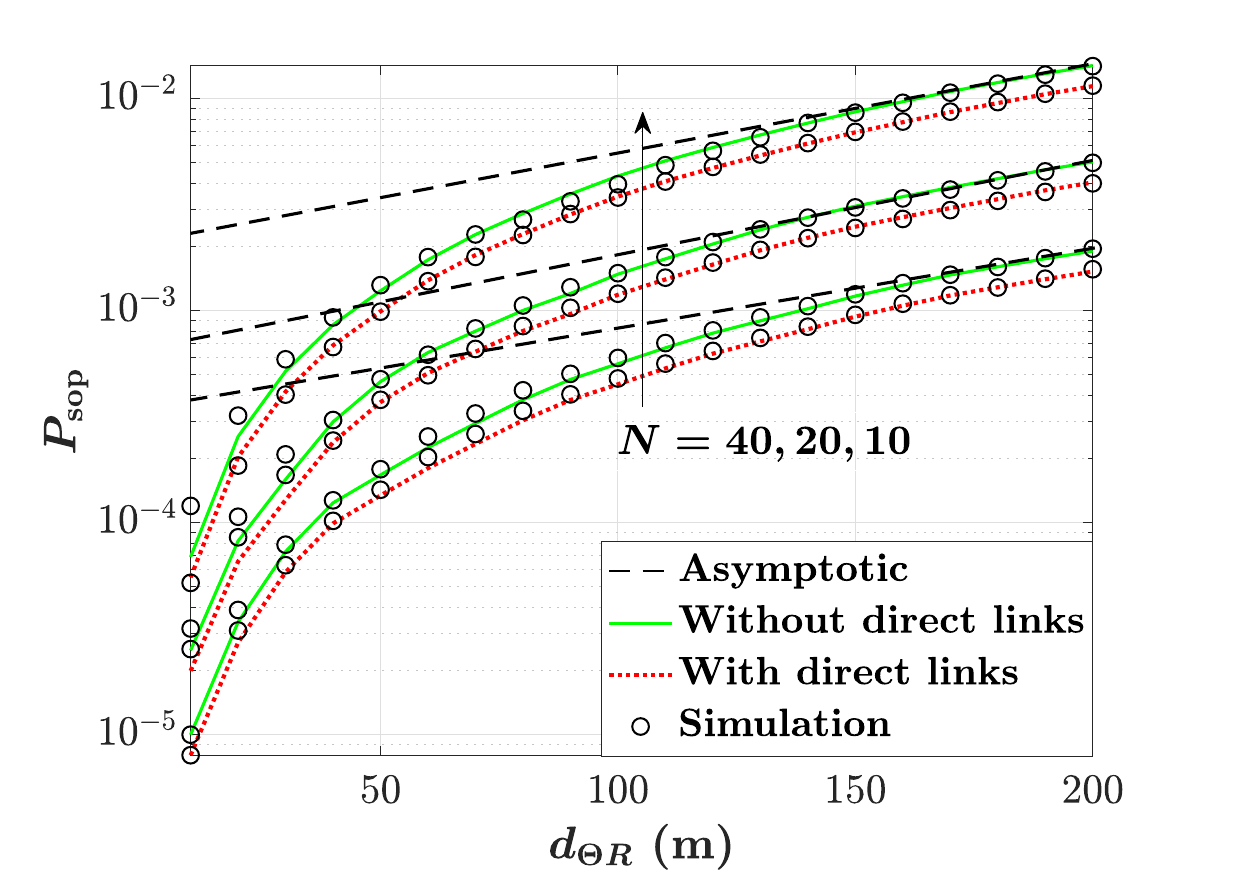}
    \caption{SOP versus $d_{\Theta R}$ for different numbers of RIS elements.}
    \label{fig:SOP_02}
\end{figure}
Fig. \ref{fig:SOP_03} depicts the performance of SOP in terms of $\bar{\gamma}_{R_2}$ for different fading parameters $m_i$ with RIS deployment. Notably, as the fading conditions become less severe (i.e., as $m_i$ increases), the SOP performance improves
due to enhanced signal quality and increased secrecy capacity. It can be also observed that, with the help of RIS, the BC system's ability to maintain secure communication even under less favorable channel conditions is significantly enhanced.
\begin{figure}[t]
    \centering    \includegraphics[width=0.35\textwidth]{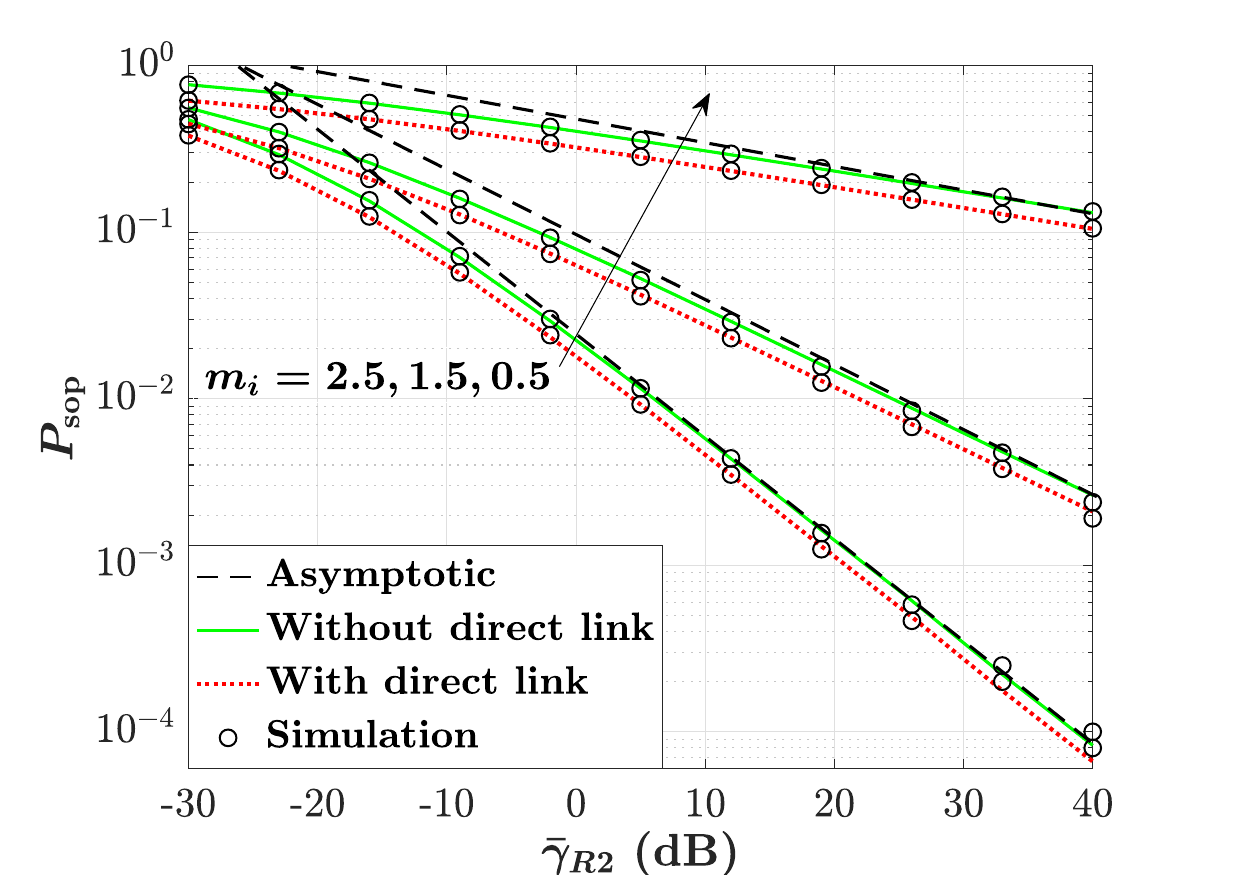}
    \caption{SOP versus $\Bar{\gamma}_{R_2}$ for different values of $m_i$.}
    \label{fig:SOP_03}
\end{figure}
Considering two scenarios (BC with and without direct links) in our analytical and simulation results reveals profound insights into the role of direct links and underscores the great potential of RIS integration to BC, where the results demonstrate a remarkable similarity in overall secrecy performance between the two scenarios. On the other hand, in cases where solely direct links exist (the scenario without RIS), the observed secrecy performance is suboptimal.
In summery, the analytical and simulation outcomes show that incorporating the RIS and increasing the number of RIS elements provide additional degrees of freedom for efficient beamforming in BC, leading to a significant enhancement in the system's secrecy performance.
%
%
\section{Conclusion}\label{conclusion}
 In this paper, we studied the secrecy performance of RIS-aided BC in terms of ASC and SOP under the Fisher-Snedecor $\mathcal{F}$ fading channels, where we considered two possible scenarios; RIS-aided BC without and with direct links.
First, for both scenarios, we derived compact analytical expressions for PDF and CDF of the SNR at both the reader and Eve. Additionally, we obtained analytical expressions of the ASC and SOP to assess the secrecy performance of the BC system under the influence of RIS. Furthermore, we provided an asymptotic analysis of the SOP and ASC to understand the system behavior in the high SNR regime.
Through Monte-Carlo simulations, we validated the analytical results, and the numerical findings demonstrated that using RIS can significantly improve the secrecy performance of BC compared to the conventional BC scenarios without RIS.
\appendices
 \section{Proof of Theorem 1}
 \label{pdf_SNR_R_01}
For computing the PDF of $X_1Y_1$, we use the distribution of the product of two random variables as
 \begin{align} \label{eq-pdf-x1x2}
f_{Y}(y)&= \int_{0}^{\infty}\int_{0}^{\infty} f_{X_1}(x_1)f_{Y_1}(y_1) \mathrm{d}x_1 \mathrm{d}y_1 , 
 \\  &=\int_{0}^{\infty} \frac{1}{\left|y_1\right|}\ f_{Y_1}(y_1)  f_{X_1}(\frac{y}{y_1}) \mathrm{d}y_1 .
\end{align}
By inserting \eqref{f_{X_1}(x_1)} and \eqref{f_{Y_1}(y_1)} in \eqref{eq-pdf-x1x2}, we get to solve 
\begin{align} \label{int_X1_Y1}
&f_{Y}(y) \hspace{-2pt}= \hspace{-3pt}\int_{0}^{\infty} {y_1}^{\frac{c-1}{2}} \hspace{2pt} \mathrm{e}^{-\frac{\sqrt{y_1}}{{\bar{y}_1}^{\frac{1}{d}}} } 
 G_{1,1}^{1,1}\left( \hspace{-4pt} \begin{array}{c}
				\frac{\lambda_{1}y}{y_1} \end{array} \hspace{-2pt}
			\Big\vert \hspace{-2pt} \begin{array}{c}
				-m_{S_{ST}}\\
				m_{ST}-1\\
			\end{array} \hspace{-4pt} \right) \mathrm{d}y_1, 
    \\ &\overset{(b)}{=} \int_{0}^{\infty} {u}^{c-2} \ \mathrm{e}^{-\frac{u}{{\bar{y}_1}^{\frac{1}{d}}} } 
 G_{1,1}^{1,1}\left( \hspace{-4pt} \begin{array}{c}
				\lambda_{1}yu^2 \end{array} \hspace{-2pt}
			\Big\vert \hspace{-2pt} \begin{array}{c}
				-m_{S_{ST}}\\
				m_{ST}-1\\
\end{array}\hspace{-4pt} \right) \mathrm{d}u,  \\&\overset{(c)}{=}   G^{1,3}_{3,1}\left( \hspace{-6pt} \begin{array}{c}	\frac{4\lambda_{1}y}{{\bar{y}_1}^{\frac{2}{d}}}   \end{array} \hspace{-2pt} 		\Big\vert \hspace{-2pt} \begin{array}{c} \frac{2-c}{2}, \frac{3-c}{2}, 2-m_{ST} \\ 1+m_{S_{ST}}\\ \end{array} \hspace{-6pt} \right),
\end{align}
where $(b)$ is followed by the integral after the variable change $u=\sqrt{y_1}$ and $(c)$ is derived from the \cite[Eq. 2.24.3.1]{ref59}. By considering $\gamma_R=\bar{\gamma}_RX_1Y_1$ and putting all the constant coefficients in \eqref{int_X1_Y1}, $f_{\gamma_R}(\gamma_R)$
is obtained as \eqref{f_{GR1}(GR1)} and the proof is completed. 
In order to compute $F_{\gamma_R}(\gamma_R)$, we use the CDF primary definition as
\begin{align} \label{cdf=pdf}
F_{\gamma_R}(\gamma_R)= \int_{0}^{\gamma_R} f_{\gamma_R}(\gamma_R) \mathrm{d}\gamma_R .
\end{align}
Next, by inserting \eqref{f_{GR1}(GR1)} into \eqref{cdf=pdf} and using \cite[Eq. 1.16.2.1]{ref59}, $F_{\gamma_R}(\gamma_R)$ is obtained as \eqref{F_{GR1}(GR1)} and the proof is completed.

 \section{Proof of Theorem 2}
 \label{pdf_SNR_E_01}
For completing the proof of Thm. \ref{theor_pdf_E}, we can take the same steps of Appendix \ref{pdf_SNR_R_01}. First, we insert \eqref{f_{Y_2}(y_2)} and \eqref{f_{X_1}(x_1)} into \eqref{eq-pdf-x1x2}. Then, by assuming $X=X_1Y_2$, we get to solve 
\begin{align} \label{int_X1_Y2}
&f_{X}(x)= \int_{0}^{\infty} \frac{1}{y_2} \ \mathrm{e}^{-\frac{y_2}{a} } 
 G_{1,1}^{1,1}\left( \hspace{-4pt} \begin{array}{c}
				\frac{\lambda_{1}x}{y_2} \end{array} \hspace{-2pt}
			\Big\vert \hspace{-2pt} \begin{array}{c}
				-m_{S_{ST}}\\
				m_{ST}-1\\
			\end{array} \hspace{-4pt} \right) \mathrm{d}y_2, 
   \\ &\overset{(d)}{=} \int_{0}^{\infty} \frac{1}{y_2} \ \mathrm{e}^{-\frac{y_2}{a} } 
 G_{1,1}^{1,1}\left( \hspace{-4pt} \begin{array}{c}
				\lambda_{1}xy_2 \end{array} \hspace{-2pt}
			\Big\vert \hspace{-2pt} \begin{array}{c}
				2-m_{ST}  \\
				1+m_{S_{ST}} \\
			\end{array} \hspace{-4pt} \right) \mathrm{d}y_2,  \\ &\overset{(c)}{=}   G^{1,2}_{2,1}\left( \hspace{-6pt} \begin{array}{c}	a\lambda_{1}x   \end{array} \hspace{-2pt} 		\Big\vert \hspace{-2pt} \begin{array}{c} 1, 2-m_{ST} \\ 1+m_{S_{ST}}\\ \end{array} \hspace{-6pt} \right),
\end{align}
where $(d)$ uses the Meijer's G-function properties in \cite[Eq. 8.2.1.14]{ref59}. 
By considering $\gamma_E=\bar{\gamma}_EX_1Y_2$ and putting all the constant coefficients in \eqref{int_X1_Y2}, $f_{\gamma_E}(\gamma_E)$
is obtained as \eqref{f_{GE1}(GE1)} and the proof is completed. Furthermore, in order to compute $F_{\gamma_E}(\gamma_E)$, we can take the same steps in Appendix \ref{pdf_SNR_R_01} and obtain the CDF of $\gamma_E$ as shown in \eqref{F_{GE1}(GE1)} by using \eqref{cdf=pdf}.

 \section{Proof of Theorem 3}
 \label{pdf_SNR_R_02}

For computing $f_{\gamma_R}(\gamma_R)$, first we obtain $M_{\gamma_{R_1}}$ and $M_{\gamma_{R_2}}$ as
\begin{align} \label{M_Gamma_R1}
M_{\gamma_{R_1}} \hspace{-4pt} &= \hspace{-2pt} \eta_1 \hspace{-3pt}\int_{0}^{\infty} \frac{\mathrm{e}^{t\gamma_{R_1}}}{\gamma_{R_1}}  
 G^{2,2}_{2,2} \hspace{-2pt} \left( \hspace{-7pt} \begin{array}{c}
\frac{\delta_1 \gamma_{R_1}}{\bar{\gamma}_{R_1}} \end{array} \hspace{-4pt}
\Bigg\vert \hspace{-3pt} \begin{array}{c} 1-m_{S_{ST}}, 1-m_{S_{TR}} \\ 
m_{ST}, m_{TR} \\
\end{array} \hspace{-6pt} \right) \hspace{-3pt} \mathrm{d}\gamma_{R_1}  \\ &\overset{(c)}{=}    
\eta_1 G^{2,3}_{3,2} \hspace{-2pt} \left( \hspace{-6pt} \begin{array}{c}
\frac{\delta_1}{s \bar{\gamma}_{R_1}} \end{array} \hspace{-4pt}
\Bigg\vert \hspace{-3pt} \begin{array}{c} 1, 1-m_{S_{ST}}, 1-m_{S_{TR}} \\ 
m_{ST}, m_{TR} \\
\end{array} \hspace{-6pt} \right) ,
\end{align}

\begin{align} \label{M_Gamma_R2}
M_{\gamma_{R_2}} \hspace{-4pt} &= \hspace{-2pt} \mathcal{G} \hspace{-3pt}\int_{0}^{\infty} \mathrm{e}^{t\gamma_{R_2}}
G^{1,3}_{3,1} \hspace{-2pt} \left( \hspace{-7pt} \begin{array}{c}
\frac{4\lambda_1 \gamma_{R_2}}{{\bar{y}_{1}}^{\frac{2}{d}}} \end{array} \hspace{-4pt}
\Bigg\vert \hspace{-3pt} \begin{array}{c} 
\frac{2-c}{2}, \frac{3-c}{2}, 2-m_{ST}  \\ 
1+m_{S_{ST}} \\
\end{array} \hspace{-6pt} \right) \hspace{-2pt} \mathrm{d}\gamma_{R_2}  \\ &\overset{(c)}{=}    
\frac{\mathcal{G} }{s} G^{1,4}_{4,1} \hspace{-2pt} \left( \hspace{-6pt} \begin{array}{c}
\frac{4\lambda_1}{s{\bar{y}_{1}}^{\frac{2}{d}}} \end{array} \hspace{-4pt}
\Bigg\vert \hspace{-3pt} \begin{array}{c} 
0, \frac{2-c}{2}, \frac{3-c}{2}, 2-m_{ST}  \\ 
1+m_{S_{ST}} \\
\end{array} \hspace{-6pt} \right) .
\end{align}
Then, by using \eqref{pdf-laplace} and the Laplace inverse formula, we have $f_{\gamma_{R}}(\gamma_{R})$ as \eqref{pdf_G_int_R},
\begin{figure*}[t]
			\normalsize
   \begin{align}  \label{pdf_G_int_R}
f_{\gamma_{R}}(\gamma_{R}) &= 
\frac{\eta_1 \mathcal{G}}{2\pi j}
\ointop_{L_{\gamma_{R}}}  \mathrm{e}^{s\gamma_{R}} \  s^{-1}   G^{2,3}_{3,2} \hspace{-2pt} \left( \hspace{-6pt} \begin{array}{c}
\frac{\delta_1}{s \bar{\gamma}_{R_1}} \end{array} \hspace{-4pt}
\Bigg\vert \hspace{-3pt} \begin{array}{c} 1, 1-m_{S_{ST}}, 1-m_{S_{TR}} \\ 
m_{ST}, m_{TR} \\
\end{array} \hspace{-6pt} \right)
G^{1,4}_{4,1} \hspace{-2pt} \left( \hspace{-6pt} \begin{array}{c}
\frac{4\lambda_1}{s{\bar{y}_{1}}^{\frac{2}{d}}} \end{array} \hspace{-4pt}
\Bigg\vert \hspace{-3pt} \begin{array}{c} 
0, \frac{2-c}{2}, \frac{3-c}{2}, 2-m_{ST}  \\ 
1+m_{S_{ST}} \\
\end{array} \hspace{-6pt} \right) 
\mathrm{d}s
\\ &\overset{(e)}{=}
\frac{\eta_1 \mathcal{G}}{\left(2\pi j\right)^3}
\underbrace{\ointop_{L_{\gamma_{R}}}
\mathrm{e}^{s\gamma_{R}} 
s^{-1-\zeta_1-\zeta_2}  \mathrm{d}s}_{I_1}
\ointop_{L_{1}} \hspace{-2pt}
\Gamma(m_{ST}-\zeta_1)  \Gamma(m_{TR}-\zeta_1) \Gamma(\zeta_1) \Gamma(m_{S_{ST}}+\zeta_1) \Gamma(m_{S_{TR}}+\zeta_1)
\left(\frac{\delta_1}{\bar{\gamma}_{R_1}}\right)^{\zeta_1}
\mathrm{d} \zeta_1
\nonumber \\ &\times
\ointop_{L_{2}} 
\Gamma(1+m_{S_{ST}}-\zeta_2)
\Gamma(1+\zeta_2) \Gamma(\frac{c}{2}+\zeta_2) \Gamma(\frac{c-1}{2}+\zeta_2) \Gamma(-1+m_{ST}+\zeta_2)
\left(\frac{4\lambda_1}{\bar{y}_{1}^{\frac{2}{d}}}\right)^{\zeta_2}  \mathrm{d} \zeta_2.
\end{align}
   \hrulefill
\vspace{-15pt}		\end{figure*}
where $(e)$ means using the integral-form demonstration of the Meijer's G-function and $L_x$ shows a specific contour. For solving $I_1$ in \eqref{pdf_G_int_R}, we use the following definition
\begin{align} \label{Gamma(w)}
   \frac{-2\pi j}{\Gamma(w)}=     \ointop_{C}  (-t)^{-w} \ \mathrm{e}^{-t} \ d_t ,
\end{align}
where, by assuming $s\gamma_R=-t$, we can compute $I_1$ as
\begin{align} \label{I_1}
    I_1=\frac{2\pi j \  {\gamma_R}^{-2-\zeta_1-\zeta_2}}{\Gamma\left(-1-\zeta_1-\zeta_2\right)}. 
\end{align}
Now, by putting \eqref{I_1} in \eqref{pdf_G_int_R}, we can re-write \eqref{pdf_G_int_R} as \eqref{pdf_H_int_R}.
\begin{figure*}[t]
			\normalsize
   \begin{align}  \label{pdf_H_int_R}
f_{\gamma_{R}}(\gamma_{R}) &= 
\frac{\eta_1 \mathcal{G} {\gamma_R}^{-2}}{\left(2\pi j\right)^2}
\ointop_{L_{1}} \ointop_{L_{2}} \hspace{-2pt}
\frac{\Gamma(m_{ST}-\zeta_1)  \Gamma(m_{TR}-\zeta_1) \Gamma(\zeta_1) \Gamma(m_{S_{ST}}+\zeta_1) \Gamma(m_{S_{TR}}+\zeta_1)}{\Gamma(-1-\zeta_1-\zeta_2)}
\left(\frac{\delta_1}{\bar{\gamma}_{R_1}\gamma_R}\right)^{\zeta_1}
\nonumber \\ &\times 
\Gamma(1+m_{S_{ST}}-\zeta_2)
\Gamma(1+\zeta_2) \Gamma(\frac{c}{2}+\zeta_2) \Gamma(\frac{c-1}{2}+\zeta_2) \Gamma(-1+m_{ST}+\zeta_2)
\left(\frac{4\lambda_1}{\bar{y}_{1}^{\frac{2}{d}}\gamma_R}\right)^{\zeta_2}  
\mathrm{d} \zeta_2 \mathrm{d} \zeta_1.
\end{align}
   \hrulefill
\vspace{-15pt}		\end{figure*}
According to the definition of bivariate Fox's H-function \cite[Eqs. 2.56-2.60]{ref60}, we can re-write \eqref{pdf_H_int_R} as \eqref{pdf_SNR_Reader}, so the proof is completed for $f_{\gamma_k}(\gamma_k)$. 

For computing $F_{\gamma_R}(\gamma_R)$, we first insert \eqref{M_Gamma_R1} and \eqref{M_Gamma_R2} into \eqref{cdf-laplace}. Then, by taking the exactly same steps of computing $f_{\gamma_R}(\gamma_R)$, $F_{\gamma_R}(\gamma_R)$ can be obtained as  \eqref{cdf_SNR_Reader}.

\section{Proof of Theorem 4}
\label{pdf_SNR_E_02}

In order to compute $f_{\gamma_E}(\gamma_E)$, first we compute $M_{\gamma_{E_1}}$ and $M_{\gamma_{E_2}}$ as
\begin{align} \label{M_Gamma_E1}
M_{\gamma_{E_1}} \hspace{-4pt} &= \hspace{-2pt} \eta_2 \hspace{-3pt}\int_{0}^{\infty} \frac{\mathrm{e}^{t\gamma_{E_1}}}{\gamma_{E_1}}  
 G^{2,2}_{2,2} \hspace{-2pt} \left( \hspace{-7pt} \begin{array}{c}
\frac{\delta_2 \gamma_{E_1}}{\bar{\gamma}_{E_1}} \end{array} \hspace{-4pt}
\Bigg\vert \hspace{-3pt} \begin{array}{c} 1-m_{S_{ST}}, 1-m_{S_{TE}} \\ 
m_{ST}, m_{TE} \\
\end{array} \hspace{-7pt} \right) \hspace{-3pt} \mathrm{d}\gamma_{E_1}  \\ &\overset{(c)}{=}    
\eta_2 G^{2,3}_{3,2} \hspace{-2pt} \left( \hspace{-6pt} \begin{array}{c}
\frac{\delta_2}{s \bar{\gamma}_{E_1}} \end{array} \hspace{-4pt}
\Bigg\vert \hspace{-3pt} \begin{array}{c} 1, 1-m_{S_{ST}}, 1-m_{S_{TE}} \\ 
m_{ST}, m_{TE} \\
\end{array} \hspace{-6pt} \right) ,
\end{align}

\begin{align} \label{M_Gamma_E2}
M_{\gamma_{E_2}} \hspace{-4pt} &= \hspace{-2pt} \frac{\mathcal{C}}{a\bar{\gamma}_{E_2}} \hspace{-3pt}\int_{0}^{\infty} \hspace{-2pt} \mathrm{e}^{t\gamma_{E_2}}
G^{1,2}_{2,1} \hspace{-2pt} \left( \hspace{-7pt} \begin{array}{c}
a \lambda_1 \gamma_{E_2} \end{array} \hspace{-5pt}
\Bigg\vert \hspace{-4pt} \begin{array}{c} 
1, 2-m_{ST}  \\ 
1+m_{S_{ST}} \\
\end{array} \hspace{-6pt} \right) \hspace{-2pt} \mathrm{d}\gamma_{E_2}  \\ &\overset{(c)}{=}    
\frac{\mathcal{C}}{a\bar{\gamma}_{E_2}s} 
G^{1,3}_{3,1} \hspace{-2pt} \left( \hspace{-6pt} \begin{array}{c}
\frac{a\lambda_1}{s} \end{array} \hspace{-4pt}
\Bigg\vert \hspace{-3pt} \begin{array}{c} 
0, 1, 2-m_{ST}  \\ 
1+m_{S_{ST}} \\
\end{array} \hspace{-6pt} \right) .
\end{align}
Then, by using \eqref{pdf-laplace} and the Laplace inverse formula, we have $f_{\gamma_{E}}(\gamma_{E})$ as \eqref{pdf_G_int_E}, 
\begin{figure*}[t]
			\normalsize
   \begin{align}  \label{pdf_G_int_E}
f_{\gamma_{E}}(\gamma_{E}) &= 
\frac{\eta_2 \mathcal{C}} {2\pi j a\bar{\gamma}_{E_2}}
\ointop_{L_{\gamma_{E}}}  \mathrm{e}^{s\gamma_{E}} \  s^{-1}   G^{2,3}_{3,2} \hspace{-2pt} \left( \hspace{-6pt} \begin{array}{c}
\frac{\delta_2}{s \bar{\gamma}_{E_1}} \end{array} \hspace{-4pt}
\Bigg\vert \hspace{-3pt} \begin{array}{c} 1, 1-m_{S_{ST}}, 1-m_{S_{TE}} \\ 
m_{ST}, m_{TE} \\
\end{array} \hspace{-6pt} \right)
G^{1,3}_{3,1} \hspace{-2pt} \left( \hspace{-6pt} \begin{array}{c}
\frac{a\lambda_1}{s} \end{array} \hspace{-4pt}
\Bigg\vert \hspace{-3pt} \begin{array}{c} 
0, 1, 2-m_{ST}  \\ 
1+m_{S_{ST}} \\
\end{array} \hspace{-6pt} \right) 
\mathrm{d}s
\\ &\overset{(e)}{=}
\frac{\eta_2 \mathcal{C}} {\left(2\pi j\right)^3 a\bar{\gamma}_{E_2}}
\underbrace{\ointop_{L_{\gamma_{E}}}
\mathrm{e}^{s\gamma_{E}} 
s^{-1-\zeta_1-\zeta_2}  \mathrm{d}s}_{I_2}
\ointop_{L_{1}} \hspace{-2pt}
\Gamma(m_{ST}-\zeta_1)  \Gamma(m_{TE}-\zeta_1) 
\Gamma(\zeta_1) \Gamma(m_{S_{ST}}+\zeta_1) \Gamma(m_{S_{TE}}+\zeta_1)
\left(\frac{\delta_2}{\bar{\gamma}_{E_1}}\right)^{\zeta_1}
\mathrm{d} \zeta_1
\nonumber \\ &\times
\ointop_{L_{2}} 
\Gamma(1+m_{S_{ST}}-\zeta_2)
\Gamma(1+\zeta_2) 
\Gamma(\zeta_2) 
\Gamma(-1+m_{ST}+\zeta_2)
\left(a\lambda_1\right)^{\zeta_2}
  \mathrm{d} \zeta_2.
\end{align}
   \hrulefill
\vspace{-15pt}		\end{figure*}
where $I_2$ can be obtained as \eqref{I_1} in exactly same way.
Then, by putting $I_2$ in \eqref{pdf_G_int_E}, we can re-write \eqref{pdf_G_int_E} as \eqref{pdf_H_int_E}, which equals to \eqref{pdf_SNR_Eve_2} based on bivariate Fox's H-function definition. Therefore, the proof is completed for $f_{\gamma_{E}}(\gamma_{E})$.
\begin{figure*}[t]
			\normalsize
   \begin{align}  \label{pdf_H_int_E}
f_{\gamma_{E}}(\gamma_{E}) &= 
\frac{\eta_2 \mathcal{C}} {\left(2\pi j\right)^2 a\bar{\gamma}_{E_2}}
\ointop_{L_{1}} \ointop_{L_{2}} \hspace{-2pt}
\frac{\Gamma(m_{ST}-\zeta_1)  \Gamma(m_{TE}-\zeta_1) \Gamma(\zeta_1) \Gamma(m_{S_{ST}}+\zeta_1) \Gamma(m_{S_{TE}}+\zeta_1)}{\Gamma(-1-\zeta_1-\zeta_2)}
\left(\frac{\delta_2}{\bar{\gamma}_{E_1}\gamma_E}\right)^{\zeta_1}
\nonumber \\ &\times 
\Gamma(1+m_{S_{ST}}-\zeta_2)
\Gamma(1+\zeta_2) 
\Gamma(\zeta_2) 
\Gamma(-1+m_{ST}+\zeta_2)
\left(a\lambda_1\right)^{\zeta_2}  
\mathrm{d} \zeta_2 \mathrm{d} \zeta_1.
\end{align}
   \hrulefill
\vspace{-15pt}		\end{figure*}

For computing $F_{\gamma_E}(\gamma_E)$, we first insert \eqref{M_Gamma_E1} and \eqref{M_Gamma_E2} into \eqref{cdf-laplace}. Then, by taking the exactly same steps of computing $f_{\gamma_E}(\gamma_E)$, $F_{\gamma_E}(\gamma_E)$ can be obtained as  \eqref{cdf_SNR_Eve_2}.

\section{Proof of Theorem 5}
\label{ASC_01}
In order to obtain ASC without considering the direct links, we can re-write \eqref{ASC_1} as \eqref{ASC_int}. 
\begin{figure*}[t]
			\normalsize
\begin{align} \label{ASC_int}
\bar{C}_s\hspace{-1 pt}= \hspace{-2 pt} \underbrace{\int_{0}^{\infty} \hspace{-2 pt} log(1+\gamma_R) f_{\gamma_R}(\gamma_R)F_{\gamma_E}(\gamma_R)  d_{\gamma_R}}_{J_1} \hspace{-1 pt}+\hspace{-3 pt} \underbrace{\int_{0}^{\infty} \hspace{-2 pt} log(1+\gamma_E) f_{\gamma_E}(\gamma_E)F_{\gamma_R}(\gamma_E)  d_{\gamma_E}}_{J_2} \hspace{-1 pt}-\hspace{-3 pt}  \underbrace{\int_{0}^{\infty} \hspace{-2 pt} log(1+\gamma_E) f_{\gamma_E}(\gamma_E)  d_{\gamma_E}}_{J_3}.
\end{align}
\hrulefill
      \vspace{-15pt}
		\end{figure*}
For solving $J_1$, we use Meijer's G-function demonstration of the logarithm function as shown in \cite[Eq. 8.4.6.5]{ref59}. Then we can re-write \eqref{ASC_int} as \eqref{ASC_J1_Meijer1}.
\begin{figure*}[t]
			\normalsize
\begin{align} \label{ASC_J1_Meijer1}
J_1 &= \frac{\mathcal{G} \mathcal{C}}{a\bar{\gamma}_{E}\ln{2}}
\int_{0}^{\infty} \hspace{-2pt} \gamma_R \hspace{2pt}
G^{1,2}_{2,2} \hspace{-2pt} \left( \hspace{-7pt} \begin{array}{c}
\gamma_R \end{array} \hspace{-4pt}
\Bigg\vert \hspace{-3pt} \begin{array}{c} 
1,1  \\ 
1,0 \\
\end{array} \hspace{-6pt} \right)
G^{1,3}_{3,1} \hspace{-2pt} \left( \hspace{-7pt} \begin{array}{c}
\frac{4\lambda_1 \gamma_{R}}{{\bar{y}_{1}}^{\frac{2}{d}}} \end{array} \hspace{-4pt}
\Bigg\vert \hspace{-3pt} \begin{array}{c} 
\frac{2-c}{2}, \frac{3-c}{2}, 2-m_{ST}  \\ 
1+m_{S_{ST}} \\
\end{array} \hspace{-6pt} \right) 
G^{1,3}_{3,2}\left( \hspace{-6pt} \begin{array}{c}	a\lambda_{1}  \gamma_R \end{array} \hspace{-2pt} 		\Big\vert \hspace{-2pt} \begin{array}{c} 0, 1,  2-m_{ST} \\ 1+m_{S_{ST}} , -1\\ \end{array} \hspace{-6pt} \right)
\hspace{-2pt} \mathrm{d}\gamma_{R}.
\end{align}
\hrulefill
\vspace{-15pt}
		\end{figure*}
For solving the integral in \eqref{ASC_J1_Meijer1}, we extend one of the Meijer's G-functions and re-write \eqref{ASC_J1_Meijer1} as \eqref{ASC_J1_Meijer2}. 
\begin{figure*}[t]
			\normalsize
\begin{align} \label{ASC_J1_Meijer2}
J_1 &= \frac{\mathcal{G} \mathcal{C}}{2\pi ja\bar{\gamma}_{E}\ln{2}}
\ointop_{L_{1}} \hspace{-2pt}
\frac{
\Gamma(1+m_{S_{TR}}-\zeta_1)
\Gamma(1+\zeta_1) \Gamma(\zeta_1) 
\Gamma(-1+m_{ST}+\zeta_1)
}
{\Gamma(2+\zeta_1)}
\left(a\lambda_1\right)^{\zeta_1}
\mathrm{d}\zeta_1
\nonumber \\ &\times
\underbrace{
\int_{0}^{\infty} \hspace{-2pt} 
{\gamma_R}^{1+\zeta_1} \hspace{2pt}
G^{1,2}_{2,2} \hspace{-2pt} \left( \hspace{-7pt} \begin{array}{c}
\gamma_R \end{array} \hspace{-4pt}
\Bigg\vert \hspace{-3pt} \begin{array}{c} 
1,1  \\ 
1,0 \\ \end{array} \hspace{-6pt} \right)
G^{1,3}_{3,1} \hspace{-2pt} \left( \hspace{-7pt} \begin{array}{c}
\frac{4\lambda_1 \gamma_{R}}{{\bar{y}_{1}}^{\frac{2}{d}}} \end{array} \hspace{-4pt}
\Bigg\vert \hspace{-3pt} \begin{array}{c} 
\frac{2-c}{2}, \frac{3-c}{2}, 2-m_{ST}  \\ 
1+m_{S_{ST}} \\
\end{array} \hspace{-6pt} \right) 
 \mathrm{d}\gamma_{R}}_{I_3}.
\end{align}
\hrulefill
\vspace{-15pt}		\end{figure*}
We can also solve $I_3$ in \eqref{ASC_J1_Meijer2}, as follows based on \cite[Eq. 2.24.1.1]{ref59}.
\begin{align} \label{I_3}
 \scalebox{0.99}{$\displaystyle I_3 \hspace{-2pt}=\hspace{-2pt}G^{3,4}_{5,3}\hspace{-2pt} \left( \hspace{-7pt} \begin{array}{c}	\frac{4\lambda_{1}}{{\bar{y}_1}^{\frac{2}{d}}}   \end{array} \hspace{-5pt} 		\Bigg\vert \hspace{-5pt} \begin{array}{c} 
 \frac{2-c}{2}, \frac{3-c}{2}, 2-m_{ST}, -2-\zeta_1, -1-\zeta_1 
 \\ 1+m_{S_{ST}}, -2-\zeta_1, -2-\zeta_1\\ \end{array} \hspace{-7pt} \right) $},
 \end{align}
Now, by plugging \eqref{I_3} (in the form of its Meijer's G-function definition) into \eqref{ASC_J1_Meijer2}, we will have $J_1$ as \eqref{ASC_01_Gamma}. According to the bivariate Fox's H-function definition, we can show \eqref{ASC_01_Gamma} as the first term of \eqref{ASC1} and the proof is completed for $J_1$. By taking the exactly same steps for computing $J_1$, we can compute $J_2$ as \eqref{ASC_02_Gamma}, which can be shown as the second term of \eqref{ASC1} based on the bivariate Fox's H-function definition. Therefore, the proof will be completed for $J_2$.
\begin{figure*}[t]
			\normalsize
\begin{align} \label{ASC_01_Gamma}
J_1&= \frac{\mathcal{G} \mathcal{C}}{a\bar{\gamma}_{E}(2\pi j)^2 }   \ointop_{L_1} \ointop_{L_2} 
\frac{\Gamma^2\left(-2-\zeta_1-\zeta_2\right) \Gamma\left(3+\zeta_1+\zeta_2\right)}{\Gamma\left(-1-\zeta_1-\zeta_2\right)} 
\frac{\Gamma\left(1+ m_{S_{ST}}-\zeta_1\right) \Gamma\left(1+\zeta_1\right) \Gamma\left(\zeta_1\right) \Gamma\left(-1+m_{ST}+\zeta_1\right)}{\Gamma\left(2+\zeta_1\right)}   
\left(a\lambda_1\right)^{\zeta_1}
\nonumber \\
 &\times 
\Gamma\left(1+ m_{S_{ST}}-\zeta_2\right) \Gamma\left(\frac{c}{2}+\zeta_2\right) \Gamma\left(\frac{c-1}{2}+\zeta_2\right) \Gamma\left(-1+m_{ST}+\zeta_2\right)
 \left(\frac{4\lambda_1}{{\bar{y}_1}^{\frac{2}{d}}} \right)^{\zeta_2} 
 \mathrm{d}{\zeta_2} \mathrm{d}{\zeta_1}.
 \end{align}
\hrulefill
\vspace{-10pt}		\end{figure*}  
\begin{figure*}[t]
			\normalsize
\begin{align} \label{ASC_02_Gamma}
J_2&= \frac{\mathcal{G} \mathcal{C}}{a\bar{\gamma}_{E}(2\pi j)^2 }   \ointop_{L_1} \ointop_{L_2} 
\frac{\Gamma\left(1+ m_{S_{ST}}-\zeta_1\right) \Gamma\left(1+\zeta_1\right) \Gamma\left(\frac{c}{2}+\zeta_1\right)  \Gamma\left(\frac{c-1}{2}+\zeta_1\right)  \Gamma\left(-1+m_{ST}+\zeta_1\right)} {\Gamma\left(2+\zeta_1\right)} \left(\frac{4\lambda_1}{{\bar{y}_1}^{\frac{2}{d}}} \right)^{\zeta_1}  
\nonumber \\
 &\times 
\Gamma\left(1+ m_{S_{ST}}-\zeta_2\right) \Gamma\left(\zeta_2\right)  \Gamma\left(-1+m_{ST}+\zeta_2\right)
 \left(a\lambda_1\right)^{\zeta_2}
 \frac{\Gamma^2\left(-2-\zeta_1-\zeta_2\right) \Gamma\left(3+\zeta_1+\zeta_2\right)}{\Gamma\left(-1-\zeta_1-\zeta_2\right)} 
 \mathrm{d}{\zeta_1} \mathrm{d}{\zeta_2} . 
 \end{align}
\hrulefill
\vspace{-15pt}		\end{figure*} 
In order to solve $J_3$, we need to solve the following integral, 
\begin{align} \label{J3_int1}
J_3 & \hspace{-2pt}=\scalebox{0.92}{$\displaystyle \hspace{-2pt}  \frac{\mathcal{A}}{\mathcal{G}} \hspace{-2pt}  \int_{0}^{\infty} \hspace{-3pt}  G^{1,2}_{2,2} \hspace{-2pt} \left( \hspace{-7pt} \begin{array}{c} \gamma_R \end{array} \hspace{-4pt} \Bigg\vert \hspace{-3pt} \begin{array}{c}  1,1  \\  1,0 \\ \end{array} \hspace{-6pt} \right) \hspace{-2pt} G^{1,2}_{2,1} \hspace{-2pt} \left( \hspace{-7pt} \begin{array}{c}	a\lambda_{1}  \gamma_E \end{array} \hspace{-4pt} 		\Big\vert \hspace{-3pt} \begin{array}{c} 
1,  2-m_{ST} \\ 1+m_{S_{ST}} \\ \end{array} \hspace{-6pt} \right)
\hspace{-2pt} \mathrm{d}\gamma_{E} $} .
\end{align}
By using \cite[Eq. 2.24.1.1]{ref59}, we can obtain \eqref{J3_int1} as the third term of \eqref{ASC1}, and thus, the proof is completed for $J_3$.
By completing the proof for $J_3$, the proof will be completed for ASC in Thm. \ref{theor_ASC_1}.

\section{Proof of Theorem 6}
\label{SOP_01}
For computing SOP, we can rewrite \eqref{SOP2} as \eqref{SOP_int1},
\begin{figure*}[t]
			\normalsize
\begin{align} \label{SOP_int1}
P_\mathrm{sop} &= \frac{\mathcal{G} \mathcal{C}}{a\bar{\gamma}_{E_2}}
\int_{0}^{\infty} \hspace{-2pt} 
\left(R_t\gamma_E+R'_t\right) 
G^{1,4}_{4,2} \hspace{-2pt} \left( \hspace{-7pt} \begin{array}{c}
\frac{4\lambda_1 \left(R_t\gamma_E+R'_t\right)}{{\bar{y}_{1}}^{\frac{2}{d}}} \end{array} \hspace{-4pt}
\Bigg\vert \hspace{-3pt} \begin{array}{c} 
0, \frac{2-c}{2}, \frac{3-c}{2}, 2-m_{ST}  \\ 
1+m_{S_{ST}}, -1 \\
\end{array} \hspace{-6pt} \right) 
G^{1,2}_{2,1}\left( \hspace{-6pt} \begin{array}{c}	a\lambda_{1}  \gamma_E \end{array} \hspace{-2pt} 		\Big\vert \hspace{-2pt} \begin{array}{c} 
1,  2-m_{ST} \\ 1+m_{S_{ST}} \\ 
\end{array} \hspace{-6pt} \right)
\hspace{-2pt} \mathrm{d}\gamma_{E}
\\ &\overset{(e)}{=}
\frac{\mathcal{G} \mathcal{C}}{a\bar{\gamma}_{E_2}\left(2\pi j\right)^2}
\underbrace{\int_{0}^{\infty} \hspace{-2pt} 
\left(R_t\gamma_E+R'_t\right)^{1+\zeta_1} {\gamma_E}^{\zeta_2} \mathrm{d}\gamma_E}_{I_4}
\ointop_{L_{2}} 
\Gamma(1+m_{S_{ST}}-\zeta_2)
\Gamma(\zeta_2) 
\Gamma(-1+m_{ST}+\zeta_2)
\left(a\lambda_1\right)^{\zeta_2}
  \mathrm{d} \zeta_2
\nonumber \\ &\times
\ointop_{L_{1}} \hspace{-2pt} 
\frac {\Gamma(1+m_{S_{ST}}-\zeta_1)
\Gamma(1+\zeta_1) \Gamma(\frac{c}{2}+\zeta_1) \Gamma(\frac{c-1}{2}+\zeta_1) \Gamma(m_{ST}+\zeta_1)
}
{\Gamma(2+\zeta_1)}
\left(\frac{4\lambda_1}{{\bar{y}_{1}}^{\frac{2}{d}}}\right)^{\zeta_1}
\mathrm{d} \zeta_1.
\end{align}
\hrulefill
\vspace{-15pt}		\end{figure*}
where $(e)$ means using the integral-form demonstration of the Meijer's G-function. In order to solve $I_4$ in \eqref{SOP_int1}, we use \cite[Eq. 3.194.3]{table_int}. Then we have
\begin{align} \label{I_4}
    I_4=
    \left(\frac{R'_t}{R_t}\right)^{1+\zeta_2}
    {R'_t}^{1+\zeta_1}
    \frac{\Gamma\left(1+\zeta_2\right)
    \Gamma\left(-2-\zeta_1-\zeta_2\right)
    }
    {\Gamma\left(-1-\zeta_1\right)}. 
\end{align}
By inserting \eqref{I_4} into \eqref{SOP_int1} and changing $\zeta_1$ to $-\zeta_1$ and $\zeta_2$ to $-\zeta_2$, we can rewrite \eqref{SOP_int1} as \eqref{SOP_01_Gamma}, which can be shown as \eqref{SOP1} based on the bivariate Fox H-function definition. Therefore, the proof will be completed for SOP in Thm. \ref{theor_SOP_1}.
\begin{figure*}[t]
			\normalsize
\begin{align} \label{SOP_01_Gamma}
P_\mathrm{sop}&= \frac{\mathcal{G} R'^2_t\mathcal{C}}{a\bar{\gamma}_{E_2}R_t(2\pi j)^2 }   \ointop_{L_1} \ointop_{L_2} \Gamma\left(-2+\zeta_1+\zeta_2\right)  \frac{\Gamma\left(1+ m_{S_{ST}}+\zeta_1\right) \Gamma\left(1-\zeta_1\right) \Gamma\left(\frac{c}{2}-\zeta_1\right) \Gamma\left(\frac{c-1}{2}-\zeta_1\right) \Gamma\left(m_{ST}-\zeta_1\right) \hspace{-2pt}
\left(\frac{{\bar{y}_1}^{\frac{2}{d}}}{4\lambda_1 R'_t} \right)^{\zeta_1}}{\Gamma\left(2-\zeta_1\right) \Gamma\left(-1+\zeta_1\right)}  \nonumber \\
 &\times \Gamma\left(1+ m_{S_{ST}}+\zeta_2\right) \Gamma\left(-\zeta_2\right) \Gamma\left(-1+m_{ST}-\zeta_2\right) \Gamma\left(1-\zeta_2\right) \left(\frac{R_t}{a\lambda_1 R'_t} \right)^{\zeta_2} \mathrm{d}{\zeta_2} \mathrm{d}{\zeta_1}.
 \end{align}
\hrulefill
\vspace{-15pt}		\end{figure*}

\section{Proof of  Theorem 7}
\label{ASC_02}
In order to obtain ASC with considering the direct links, we can follow the same steps we took in Appendix \ref{ASC_01}. For this, we start with computing $J_1$ in \eqref{ASC_int}. By representing the Meijer G demonstration of the logarithm function and using the integral-form demonstration of bivariate Fox H-function in \eqref{pdf_SNR_Reader} and \eqref{cdf_SNR_Eve_2}, we can write $J_1$ as \eqref{J1_ASC2_int1}. 
\begin{figure*}[t]
			\normalsize
   \begin{align}  \label{J1_ASC2_int1}
&J_1= \frac{\eta_1\eta_2\mathcal{G}\mathcal{C}}{a\bar{\gamma}_{E_2}ln(2) \left(2\pi j\right)^4} 
\underbrace{
\int_{0}^{\infty} \hspace{-2pt} 
{\gamma_R}^{-5-\zeta_1-\zeta_2-\zeta_3-\zeta_4} \hspace{2pt}
G^{1,2}_{2,2} \hspace{-2pt} \left( \hspace{-7pt} \begin{array}{c}
\gamma_R \end{array} \hspace{-4pt}
\Bigg\vert \hspace{-3pt} \begin{array}{c} 
1,1  \\ 
1,0 \\ \end{array} \hspace{-6pt} \right) 
 \mathrm{d}\gamma_{R}}_{I_5}   
\nonumber  \\ &\times \hspace{-4pt} 
\ointop_{L_{1}} \hspace{-2pt}
\ointop_{L_{2}}\hspace{-2pt}
\scalebox{0.85}{$\displaystyle
\frac{\Gamma(m_{ST}\hspace{-2pt}-\hspace{-2pt}\zeta_1)  \Gamma(m_{TR}\hspace{-2pt}-\hspace{-2pt}\zeta_1) \Gamma(\zeta_1) \Gamma(m_{S_{ST}}\hspace{-2pt}+\hspace{-2pt}\zeta_1) \Gamma(m_{S_{TR}}\hspace{-2pt}+\hspace{-2pt}\zeta_1) \hspace{-3pt}
\left(\hspace{-2pt}\frac{\delta_1}{\bar{\gamma}_{R_1}}\hspace{-2pt}\right)^{\zeta_1} \hspace{-2pt}
\Gamma(1\hspace{-2pt}+\hspace{-2pt}m_{S_{ST}}\hspace{-2pt}-\hspace{-2pt}\zeta_2) \Gamma(1\hspace{-2pt}+\hspace{-2pt}\zeta_2) \Gamma(\frac{c}{2}\hspace{-2pt}+\hspace{-2pt}\zeta_2) \Gamma(\frac{c-1}{2}\hspace{-2pt}+\hspace{-2pt}\zeta_2) \Gamma(m_{ST}\hspace{-2pt}-\hspace{-2pt}1\hspace{-2pt}+\hspace{-2pt}\zeta_2) \hspace{-3pt}
\left(\hspace{-2pt}\frac{4\lambda_1}{\bar{y}_{1}^{\frac{2}{d}}}\hspace{-3pt}\right)^{\zeta_2}}
{\Gamma(-1-\zeta_1-\zeta_2)}
\mathrm{d}{\zeta_2} \mathrm{d}{\zeta_1}
$}
   \nonumber  \\  &\times 
    \hspace{-4pt} 
\ointop_{L_{3}} \hspace{-2pt}
\ointop_{L_{4}}\hspace{-2pt}
\scalebox{0.88}{$\displaystyle
\frac{\Gamma(m_{ST}\hspace{-2pt}-\hspace{-2pt}\zeta_3)  \Gamma(m_{TE}\hspace{-2pt}-\hspace{-2pt}\zeta_3) \Gamma(\zeta_3) \Gamma(m_{S_{ST}}\hspace{-2pt}+\hspace{-2pt}\zeta_3) \Gamma(m_{S_{TE}}\hspace{-2pt}+\hspace{-2pt}\zeta_3) \hspace{-3pt}
\left(\hspace{-2pt}\frac{\delta_2}{\bar{\gamma}_{E_1}}\hspace{-2pt}\right)^{\zeta_3} \hspace{-2pt}
\Gamma(1\hspace{-2pt}+\hspace{-2pt}m_{S_{ST}}\hspace{-2pt}-\hspace{-2pt}\zeta_4) \Gamma(1\hspace{-2pt}+\hspace{-2pt}\zeta_4) \Gamma(\zeta_4) \Gamma(m_{ST}\hspace{-2pt}-\hspace{-2pt}1\hspace{-2pt}+\hspace{-2pt}\zeta_4) \hspace{-2pt}
\left(\hspace{-1pt}a\lambda_1\hspace{-1pt}\right)^{\zeta_4}}
{\Gamma(-2-\zeta_3-\zeta_4)}
\mathrm{d}{\zeta_4} \mathrm{d}{\zeta_3}
$}.
\end{align}
   \hrulefill
\vspace{-10pt}		\end{figure*}
In order to solve $I_5$ in \eqref{J1_ASC2_int1}, we use \cite[Eq. 2.24.2.1]{ref59}. Then we have
\begin{align} \label{I_5}
 I_5 \hspace{-2pt}=\hspace{-2pt} \scalebox{0.899}{$\displaystyle  
    \frac{\Gamma\left(5+\zeta_1+\zeta_2+\zeta_3+\zeta_4\right)
    \Gamma\left(-1-\zeta_1-\zeta_2\right)
    \Gamma\left(-2-\zeta_3-\zeta_4\right)
    }
    {\Gamma^2\left(4+\zeta_1+\zeta_2+\zeta_3+\zeta_4\right)
    \Gamma\left(-3-\zeta_1-\zeta_2-\zeta_3-\zeta_4\right)
    }  $}. 
\end{align}
Now, by inserting \eqref{I_5} into \eqref{J1_ASC2_int1} and then using the multivariate Fox H-function definition \cite{ref61}, we can rewrite $J_1$ as shown in the first item of \eqref{ASC02}. 
By taking the exactly same steps for computing $J_1$, we can compute $J_2$ as shown in the second item of \eqref{ASC02}. 
In order to solve $J_3$, we use the Meijer G demonstration of the logarithm function and the integral-form demonstration of bivariate Fox H-function in \eqref{pdf_SNR_Eve_2}. Therefore, we can write $J_3$ as \eqref{J3_ASC2_int1}. 
\begin{figure*}[t]
			\normalsize
   \begin{align}  \label{J3_ASC2_int1}
&J_3= \frac{\eta_2\mathcal{C}}{a\bar{\gamma}_{E_2}ln(2) \left(2\pi j\right)^4} 
\underbrace{
\int_{0}^{\infty} \hspace{-2pt} 
{\gamma_E}^{-2-\zeta_1-\zeta_2} \hspace{2pt}
G^{1,2}_{2,2} \hspace{-2pt} \left( \hspace{-7pt} \begin{array}{c}
\gamma_E \end{array} \hspace{-4pt}
\Bigg\vert \hspace{-3pt} \begin{array}{c} 
1,1  \\ 
1,0 \\ \end{array} \hspace{-6pt} \right) 
 \mathrm{d}\gamma_{E}}_{I_6}   
\\ \nonumber  &\times \hspace{-4pt} 
\ointop_{L_{1}} \hspace{-2pt}
\ointop_{L_{2}}\hspace{-2pt}
\scalebox{0.95}{$\displaystyle
\frac{\Gamma(m_{ST}\hspace{-2pt}-\hspace{-2pt}\zeta_1)  \Gamma(m_{TE}\hspace{-2pt}-\hspace{-2pt}\zeta_1) \Gamma(\zeta_1) \Gamma(m_{S_{ST}}\hspace{-2pt}+\hspace{-2pt}\zeta_1) \Gamma(m_{S_{TE}}\hspace{-2pt}+\hspace{-2pt}\zeta_1) 
\hspace{-3pt} 
\left(\hspace{-2pt}\frac{\delta_2}{\bar{\gamma}_{E_1}}\hspace{-2pt}\right)^{\zeta_1} 
\hspace{-2pt}
\Gamma(1\hspace{-2pt}+\hspace{-2pt}m_{S_{ST}}\hspace{-2pt}-\hspace{-2pt}\zeta_2) \Gamma(1\hspace{-2pt}+\hspace{-2pt}\zeta_2)  \Gamma(\zeta_2) \Gamma(m_{ST}\hspace{-2pt}-\hspace{-2pt}1\hspace{-2pt}+\hspace{-2pt}\zeta_2) 
\hspace{-3pt}
\left(\hspace{-1pt}a\lambda_1\hspace{-2pt}\right)^{\zeta_2}}
{\Gamma(-1-\zeta_1-\zeta_2)}
\mathrm{d}{\zeta_2} \mathrm{d}{\zeta_1}
$}.
\end{align}
   \hrulefill
\vspace{-15pt}		\end{figure*}
In order to solve $I_6$ in \eqref{J3_ASC2_int1}, we use \cite[Eq. 2.24.2.1]{ref59}. Then we have
\begin{align} \label{I_6}
 I_6 = \scalebox{0.999}{$\displaystyle  
    \frac{
    \Gamma\left(2+\zeta_1+\zeta_2\right)
    }
    {\Gamma^2\left(1+\zeta_1+\zeta_2\right)
    \Gamma\left(-\zeta_1-\zeta_2\right)
    }  $}. 
\end{align}
Now, by inserting \eqref{I_6} into \eqref{J3_ASC2_int1} and then using the multivariate Fox H-function definition \cite{ref61}, we can rewrite $J_3$ as shown in the third item of \eqref{ASC02}. Therefore, the proof will be completed for ASC in Thm. \ref{theor_ASC2}.

\section{Proof of Theorem 8}
\label{SOP_02}
By inserting the integral-form demonstration of bivariate Fox H-function of of \eqref{cdf_SNR_Reader} and \eqref{pdf_SNR_Eve_2} into \eqref{SOP2}, we can rewrite SOP as \eqref{SOP2_int1}. 
\begin{figure*}[t]
			\normalsize
   \begin{align}  \label{SOP2_int1}
&P_\mathrm{sop}= \frac{\eta_1\eta_2\mathcal{G}\mathcal{C}}{a\bar{\gamma}_{E_2}ln(2) \left(2\pi j\right)^4} 
\underbrace{\int_{0}^{\infty} \hspace{-2pt} 
\left(R_t\gamma_E+R'_t\right)^{-3-\zeta_1-\zeta_2} {\gamma_E}^{-2-\zeta3-\zeta_4} \mathrm{d}\gamma_E}_{I_7}   
\nonumber  \\ &\times \hspace{-4pt} 
\ointop_{L_{1}} \hspace{-2pt}
\ointop_{L_{2}}\hspace{-2pt}
\scalebox{0.85}{$\displaystyle
\frac{
\Gamma(m_{ST}\hspace{-2pt}-\hspace{-2pt}\zeta_1)  
\Gamma(m_{TR}\hspace{-2pt}-\hspace{-2pt}\zeta_1) 
\Gamma(\zeta_1) 
\Gamma(m_{S_{ST}}\hspace{-2pt}+\hspace{-2pt}\zeta_1) 
\Gamma(m_{S_{TR}}\hspace{-2pt}+\hspace{-2pt}\zeta_1) 
\hspace{-3pt}
\left(\hspace{-2pt}\frac{\delta_1}{\bar{\gamma}_{R_1}}\hspace{-2pt}\right)^{\zeta_1} 
\hspace{-2pt}
\Gamma(1\hspace{-2pt}+\hspace{-2pt}m_{S_{ST}}\hspace{-2pt}-\hspace{-2pt}\zeta_2) 
\Gamma(1\hspace{-2pt}+\hspace{-2pt}\zeta_2) \Gamma(\frac{c}{2}\hspace{-2pt}+\hspace{-2pt}\zeta_2) \Gamma(\frac{c-1}{2}\hspace{-2pt}+\hspace{-2pt}\zeta_2) \Gamma(m_{ST}\hspace{-2pt}-\hspace{-2pt}1\hspace{-2pt}+\hspace{-2pt}\zeta_2) \hspace{-3pt}
\left(\hspace{-2pt}\frac{4\lambda_1}{\bar{y}_{1}^{\frac{2}{d}}}\hspace{-3pt}\right)^{\zeta_2}}
{\Gamma(-2-\zeta_1-\zeta_2)}
\mathrm{d}{\zeta_2} \mathrm{d}{\zeta_1}
$}
   \nonumber  \\  &\times 
    \hspace{-4pt} 
\ointop_{L_{3}} \hspace{-2pt}
\ointop_{L_{4}}\hspace{-2pt}
\scalebox{0.88}{$\displaystyle
\frac{\Gamma(m_{ST}\hspace{-2pt}-\hspace{-2pt}\zeta_3)  \Gamma(m_{TE}\hspace{-2pt}-\hspace{-2pt}\zeta_3) \Gamma(\zeta_3) \Gamma(m_{S_{ST}}\hspace{-2pt}+\hspace{-2pt}\zeta_3) \Gamma(m_{S_{TE}}\hspace{-2pt}+\hspace{-2pt}\zeta_3) \hspace{-3pt}
\left(\hspace{-2pt}\frac{\delta_2}{\bar{\gamma}_{E_1}}\hspace{-2pt}\right)^{\zeta_3} \hspace{-2pt}
\Gamma(1\hspace{-2pt}+\hspace{-2pt}m_{S_{ST}}\hspace{-2pt}-\hspace{-2pt}\zeta_4) \Gamma(1\hspace{-2pt}+\hspace{-2pt}\zeta_4) \Gamma(\zeta_4) \Gamma(m_{ST}\hspace{-2pt}-\hspace{-2pt}1\hspace{-2pt}+\hspace{-2pt}\zeta_4) \hspace{-2pt}
\left(\hspace{-1pt}a\lambda_1\hspace{-1pt}\right)^{\zeta_4}}
{\Gamma(-1-\zeta_3-\zeta_4)}
\mathrm{d}{\zeta_4} \mathrm{d}{\zeta_3}
$}.
\end{align}
   \hrulefill
\vspace{-15pt}		\end{figure*}
In order to solve $I_7$ in \eqref{SOP2_int1}, we use \cite[Eq. 3.194.3]{table_int}. Then we have
\begin{align} \label{I_7}
   \scalebox{0.88}{$\displaystyle
   I_7 \hspace{-2pt}=\hspace{-2pt}
    \left(\frac{R_t}{R'_t}\right)^{1+\zeta_3+\zeta_4} \hspace{-2pt}
    {R'_t}^{-3-\zeta_1-\zeta_2}
    \frac{\Gamma\hspace{-2pt}\left(-2\hspace{-2pt}-\hspace{-2pt}\zeta_1\hspace{-2pt}-\hspace{-2pt}\zeta_2\right)\Gamma\hspace{-2pt}\left(-1\hspace{-2pt}-\hspace{-2pt}\zeta_3\hspace{-2pt}-\hspace{-2pt}\zeta_4\right)}
    {\Gamma\left(4+\zeta_1+\zeta_2+\zeta_3+\zeta_4\right)}
    $}.
\end{align}
By inserting \eqref{I_7} into \eqref{SOP2_int1} and then using the multivariate Fox H-function definition \cite{ref61}, we can represent \eqref{SOP2_int1} as shown in \eqref{SOP2}, thus, the proof is completed for SOP in Thm. \ref{theor_SOP2}.

\section{Proof of Proposition 1}
\label{asy_ASC_01_proof}

In the case of $\bar{\gamma}_R\rightarrow\infty$, the asymptotic ASC is defined
as
\begin{align} \label{ASC_asy_eq_J}
\bar{C}_{\mathrm{s}}^{\mathrm{asy}}= J'_1+J'_2-J_3 .
\end{align}
In order to compute $J'_1$, we can see that the bivariate Fox's H-function in
\eqref{ASC_01_Gamma} is evaluated at the highest poles on the left of $L_2$, i.e., $\zeta_2=-2-\zeta_1$. Thus, we have the integral in \eqref{proof_r1} for the counter $L_2$.
\begin{figure*}[t]
			\normalsize
\begin{align}  \label{proof_r1}
&\scalebox{0.95}{$\displaystyle\mathcal{R}_1 \hspace{-2pt}=\hspace{-2pt} \frac{1}{2\pi j}\ointop_{L_2}
\frac{\Gamma^2\left(-2-\zeta_1-\zeta_2\right) \Gamma\left(3+\zeta_1+\zeta_2\right)}{\Gamma\left(-1-\zeta_1-\zeta_2\right)}
 \underset{\xi(\zeta_2)}{\underbrace{\Gamma\hspace{-2pt} \left(1+ m_{S_{ST}}-\zeta_2\right) \Gamma\hspace{-2pt}\left(\frac{c}{2}+\zeta_2\right) \Gamma\hspace{-2pt}\left(\frac{c-1}{2}+\zeta_2\right) \Gamma\hspace{-2pt}\left(-1+m_{ST}+\zeta_2\right)
 \hspace{-2pt}\left(\frac{4\lambda_1}{{\bar{y}_1}^{\frac{2}{d}}} \right)^{\zeta_2}}}$}.
\end{align}
\hrulefill
\vspace{-15pt}		\end{figure*}
Since $\mathcal{R}_1=\mathrm{Res}\left[\xi(\zeta_2),-2-\zeta_1\right]$, we can rewrite \eqref{proof_r1} as
\begin{align} \label{Res1_ASC_J1}
\mathcal{R}_1
&=\lim_{\zeta_2\rightarrow -2-\zeta_1}\left(-2-\zeta_1-\zeta_2\right)\xi(\zeta_2)\\ \nonumber  \label{Res2_ASC_J1}
&=\Gamma\hspace{-2pt} \left(3+ m_{S_{ST}}+\zeta_1\right) \Gamma\hspace{-2pt}\left(\frac{c-4}{2}-\zeta_1\right) \Gamma\hspace{-2pt}\left(\frac{c-5}{2}-\zeta_1\right) 
\nonumber \\ &\times
\Gamma\hspace{-2pt}\left(-3+m_{ST}-\zeta_1\right)
 \hspace{-2pt}\left(\frac{4\lambda_1}{{\bar{y}_1}^{\frac{2}{d}}} \right)^{-2-\zeta_1}
\end{align}
Now, by inserting \eqref{Res2_ASC_J1} into \eqref{proof_r1}, $J'_1$ can be determined as \eqref{J1_asy_int_1}
\begin{figure*}[t]
			\normalsize
\begin{align} \label{J1_asy_int_1}
J'_1&= \frac{\mathcal{G} \mathcal{C}{\bar{y}_1}^{\frac{4}{d}}}{16a\lambda_1^2\bar{\gamma}_{E} 2\pi j }   \ointop_{L_1} 
\frac{\Gamma\left(1+ m_{S_{ST}}-\zeta_1\right) \Gamma\left(1+\zeta_1\right) \Gamma\left(\zeta_1\right) \Gamma\left(-1+m_{ST}+\zeta_1\right)}{\Gamma\left(2+\zeta_1\right)} 
\nonumber \\ &\times 
\Gamma \left(3+ m_{S_{ST}}+\zeta_1\right) \Gamma\left(\frac{c-4}{2}-\zeta_1\right) \Gamma\left(\frac{c-5}{2}-\zeta_1\right)
\Gamma\left(-3+m_{ST}-\zeta_1\right)
\left(\frac{{\bar{y}_1}^{\frac{2}{d}}a}{4} \right)^{\zeta_1}
\mathrm{d}{\zeta_1}.
 \end{align}
 \hrulefill
\vspace{-10pt}		\end{figure*}
 where, by using the definition of the univariate Meijer's G-function, the proof is completed for $J_1$ (the first term of \eqref{asy_ASC123}).
 Similarly, $J'_2$ i.e., the secoend term in \eqref{asy_ASC123} can be obtained by computing the bivariate
Fox’s H-function in \eqref{ASC_02_Gamma} at the pole $\zeta_2=-2-\zeta_1$.
Finally, by plugging $J_3$ in the third term of \eqref{ASC1} into \eqref{asy_ASC123}, the proof of the asymptotic
ASC is completed.

\section{Proof of Proposition 2}
\label{asy_SOP_01_proof}

In the case of $\bar{\gamma}_R\rightarrow\infty$, the bivariate Fox's H-function in
\eqref{SOP_01_Gamma} is evaluated at the highest poles on the left of $L_2$, i.e., $\zeta_2=2-\zeta_1$. Thus, we have the integral in \eqref{proof-r2} for the counter $L_2$.
\begin{figure*}[t]
			\normalsize
\begin{align} 
&\mathcal{R}_2= \frac{1}{2\pi j}\ointop_{L_2}\Gamma\left(-2+\zeta_1+\zeta_2\right)
 \underset{\chi(\zeta_2)}{\underbrace{\Gamma\left(-\zeta_2\right) \Gamma\left(-1+m_{ST}-\zeta_2\right) \Gamma\left(1-\zeta_2\right)\left(1+ m_{S_{ST}}+\zeta_2\right)\left(\frac{R_t}{a\lambda_1 R'_t} \right)^{\zeta_2} \mathrm{d}{\zeta_2}}}.
\label{proof-r2}
\end{align}
\hrulefill
\vspace{-10pt}		\end{figure*}
Since $\mathcal{R}_2=\mathrm{Res}\left[\chi(\zeta_2),2-\zeta_1\right]$, we can rewrite \eqref{proof-r2} as
\begin{align} \label{Res_SOP}
\mathcal{R}_2
&=\lim_{\zeta_2\rightarrow 2-\zeta_1}\left(-2+\zeta_1+\zeta_2\right)\chi(\zeta_2)\\ \nonumber  \label{Res_SOP2}
&=\Gamma\left(\zeta_1-2\right)\Gamma\left(\zeta_1+m_{ST}-3\right)\Gamma\left(\zeta_1-1\right)\\
&\times\Gamma\left(3+m_{S_{ST}}-\zeta_1\right)\left(\frac{R_t}{a\lambda_1 R_t'}\right)^{2-\zeta_1}.
\end{align}
Now, by inserting \eqref{Res_SOP2} into \eqref{SOP_01_Gamma}, the asymptotic SOP can be determined as
\begin{align} \nonumber
&P_\mathrm{sop}^\mathrm{asy}= \frac{\mathcal{G} \mathcal{C}R_t}{a^3\lambda_1^2\bar{\gamma}_{E_2} 2\pi j }   \ointop_{L_1} \Gamma\left(1+ m_{S_{ST}}+\zeta_1\right) \left(\frac{{\bar{y}_1}^{\frac{2}{d}}a}{4R_t} \right)^{\zeta_1}\\\nonumber
&\times\frac{ \Gamma\left(\frac{c}{2}-\zeta_1\right) \Gamma\left(\frac{c-1}{2}-\zeta_1\right) \Gamma\left(m_{ST}-\zeta_1\right)\Gamma\left(1-\zeta_1\right) }{\Gamma\left(2-\zeta_1\right)}\\
&\times \scalebox{0.98}{$\displaystyle \Gamma\left(\zeta_1-2\right)\Gamma\left(\zeta_1+m_{ST}-3\right)\Gamma\left(3+m_{S_{ST}}-\zeta_1\right)\mathrm{d}{\zeta_1} $},
 \end{align}
 where, by using the definition of the univariate Meijer's G-function, the proof is completed.

\end{document}